\documentclass[12pt,reqno]{amsart}

\usepackage{graphicx}
\usepackage{amsthm,amsmath,amsfonts,amssymb,amsxtra,appendix,bookmark, color}
\usepackage{cite}

\newtheorem{theorem}{Theorem}[section]
\newtheorem{lemma}[theorem]{Lemma}

\theoremstyle{definition}

\newtheorem{remark}[theorem]{Remark}

\newtheorem{proposition}[theorem]{Proposition}

\numberwithin{equation}{section}  

\begin{document}

\title[Spontaneous mass generation in a lattice NJL model]{Spontaneous mass generation and chiral symmetry breaking 
in a lattice Nambu--Jona-Lasinio model}

\author[Y. Goto]{Yukimi Goto\textsuperscript{1} }
\thanks{\textsuperscript{1}Kyushu University,~Faculty of Mathematics,~Nishi-ku,~Fukuoka~819-0395, Japan, Email:~{\tt   yukimi@math.kyushu-u.ac.jp}}

\author[T. Koma]{Tohru Koma\textsuperscript{2} }
\thanks{\textsuperscript{2}Gakushuin University (retired), Department of Physics, Mejiro, Toshima-ku, Tokyo 171-8588, Japan}

\maketitle

\medskip

\noindent
{\bf Abstract:} 
We study a lattice Nambu--Jona-Lasinio model with interacting staggered fermions in the Kogut--Susskind Hamiltonian formalism. 
The model has a discrete chiral symmetry but not the usual continuous chiral symmetry. 
In a strong coupling regime for the four-fermion interaction, we prove that the mass of the fermions is spontaneously generated  
at sufficiently low non-zero temperatures in the dimensions $\nu \ge 3$ of the model and zero temperature in $\nu \ge 2$. 
Due to the phase transition, the discrete chiral symmetry of the model is broken.
Our analysis is based on the reflection positivity for fermions and the method of the infrared bound.
\bigskip

 \section{Introduction}
The strongly interacting elementary particles have been historically called hadrons, 
which nowadays are thought to be the bound states of quarks. 
One of the main goals of particle physics is to explain the properties of hadrons by starting from quantum chromodynamics (QCD), 
which describes the dynamics of quarks and gluons (gauge bosons).
The masses of the {hadrons} are believed to be generated by the interaction between 
the quarks and the gluons, accompanied by the breakdown of the chirality of the massless quarks. 
Since the chiral symmetry is continuous, the so-called Nambu--Goldstone bosons are expected to appear \cite{Nambu,Goldstone,GSW}. 
These massless bosons are thought to be the mesons such as the $\pi$-meson, 
since the pion exhibits the very small experimental mass compared to other hadrons.

In order to deal with the chiral symmetry breaking, Nambu and Jona-Lasinio \cite{NJL,NJL2} proposed a fermion model 
with an effective four-fermion interaction 
which is induced by the original fermion-boson interaction, by relying on an interesting analogy to 
Bardeen--Cooper--Schrieffer (BCS) \cite{BCS} theory for superconductivity. 
Nowadays, it is called Nambu--Jona-Lasinio (NJL) model \cite{Hatsuda}, which has the chiral symmetry for the fermions. 


Since mathematically rigorous results are still rare for QCD, we should remark that Salmhofer and Seiler \cite{SS1,SS2} 
proved the chiral symmetry breaking for staggered lattice fermions \cite{Susskind,BRSKJSS} coupled to gauge fields in 
the strong coupling limit for the gauge fields. Although their results are restricted to the strong coupling limit, 
their proof shows that the effective four-fermion interactions play a crucial role for the emergence of 
the symmetry breaking states as Nambu pointed out \cite{Nambu}.

On the other hand, in condensed matter physics, a strong two-body fermion interaction such as Coulomb interactions 
is well known to induce a metal-insulator transition, which is called Mott transition. 
Namely, due to the strong interaction, there appears a mass gap above the Fermi level.   
In a mathematically rigorous manner, Lemberger and Macris~\cite{LM} dealt with a fermion model with 
a strong repulsive Coulomb interaction and proved 
the existence of a charge long-range order for a certain range of non-zero temperatures in the dimensions $\nu \ge 2$ 
by using the Peierls argument. Their results stimulated studies of fermion systems. 
Actually, the cluster expansion method developed by Kennedy and Tasaki \cite{KennedyTasaki} was extended to 
fermion systems \cite{DFF1, DFF2, Ueltschi, FBU}. Unfortunately, the method suffers from the fermion sign problem 
for treating the ground states.

In this paper, we study a staggered lattice fermion model \cite{Susskind,BRSKJSS} with a four-fermion interaction 
in the Kogut--Susskind Hamiltonian formalism \cite{KS} in the dimensions $\nu\ge 2$. 
The model has a discrete chiral $\mathbb{Z}_2$ symmetry\footnote{{This chiral $\mathbb{Z}_2$ symmetry 
seems to be justified only 
in the low energy regime. See Appendices~\ref{staggeredfermion} and \ref{Sec:chiralint} for the details.}} 
when the mass parameter of the mass term in the Hamiltonian 
is equal to zero. Conversely, when the mass term is nonvanishing, the model does not have the chiral symmetry.  
Thus, the mass term itself is the order parameter for the phase transition.  
We prove that the phase transition occurs for the strong coupling of the four-fermion interaction at zero temperature 
in the dimensions $\nu\ge 2$ and low temperatures in the dimensions $\nu\ge 3$. 
In other words, there occurs spontaneous mass generation in the sense that the expectation value of the mass term 
takes non-trivial values for the symmetry breaking states in the infinite-volume limit. 
In order to prove these statements, we use the reflection positivity for fermion systems \cite{Koma4} 
which is an extension of the reflection positivity for spin systems \cite{DLS}. 
In fact, the staggered fermion model has the reflection positivity. 


We should remark on the pros and cons in the staggered lattice fermion formalism. 
As is well-known~\cite{Rothe}, in the staggered fermion formalism, there generally appear multiple flavors of the massless Dirac fermions, 
and they cannot be reduced to a single flavor. This is nothing but the fermion doubling problem. 
Clearly, there arises the problem of deriving the correct anomaly for a single flavor. 
In order to solve this problem, the so-called rooting procedure was introduced into QCD computations~\cite{Kronfeld}. 
However, this method seems to be still controversial.

{Another way to avoid multiple flavors is to introduce an extra term, such as the Wilson term, into the system 
so that it changes undesired massless Dirac fermions to massive ones. 
Although Wilson's formulation eliminates the fermion doubling, the extra term breaks explicitly the chiral symmetry 
even in the massless case. This lattice effect is stronger than that of the staggered fermion formalism. 
Therefore, the advantage of using the staggered fermions is that one can easily study the mass generation, 
in contrast to the Wilson fermions. In particular, it enables us to use non-perturbative approaches as in the present paper. 
On the other hand, it is known that the Wilson fermions reproduce the above chiral anomaly 
in the continuum limit~\cite{KSm}.
}


Our model can be interpreted also as a model in condensed matter physics. 
As is well known, the staggered lattice fermion model is equivalent to the fermion model having 
$\pi$ flux \cite{LiebFlux,LN,MN} for all the unit plaquettes.\footnote{See, e.g., \cite{Koma5}.} 
Besides, the strong four-fermion (Coulomb) interaction induces a charge density wave 
and the particle-hole symmetry at the half-filling is spontaneously broken in the infinite-volume limit. 
This is nothing but Mott transition, i.e., a metal-insulator transition. 
Unfortunately, we cannot prove the existence of the mass gap above the ground states, 
i.e., the change of the massless fermions to the massive ones. 
In passing, the charge density operator corresponds to the mass term of the fermion Hamiltonian in QCD.  

{In the case of the weak coupling regime, one can treat the four-fermion interactions as a perturbation for 
the free fermion Hamiltonian. Clearly, the unperturbed ground state is given by the unique Dirac sea.
Let us require that the perturbative series is convergent with a convergence radius. 
Then, the uniqueness of the ground state implies that the expectation value of the order parameter is order by order vanishing. 
Therefore, a non-vanishing result can be found only outside the convergence radius.} 


In order to take the continuum limit of our model, 
we recall a standard procedure which relies on the critical behavior 
of the phase transition. To use this procedure, we require that 
the mass which is determined by the two-point correlation 
approaches to zero when the coupling constant $g$ of the four-fermion interaction 
approaches to the critical value $g_c$ in the Mott order phase. 
Clearly, this mass is conceptually different from the expectation value  
{of the mass observable} in the above explanation of the mass generation for {fermions}. 
Under this assumption, we vary the distance scale in the two-point correlation 
so that the mass times the length scale is kept to a non-vanishing finite value 
as $g\rightarrow g_c$. {In addition, we may have to also control the mass parameter $m$ 
so that the expectation value {of the mass observable}  
is nonvanishing in the continuum limit.\footnote{{It is very hard to justify these assumptions about the continuum limit.
In fact, that is not the aim of the present paper.}}}
Since there exists a weak$^\ast$-limit for any sequence of states, 
one can construct the continuum limit by this procedure. 
However, a reader might think that there is no such a continuum limit 
because Nambu--Jona-Lasinio model is well known to be unrenormalizable. 
Of course, the state so obtained may be an undesirable one, e.g., 
showing non-interacting fermions. But, if we can obtain a continuum limit 
even in this trivial case for the interaction, then the chiral symmetry of the fermions 
is broken in the sense of the mass generation for the usual Dirac fermions. 
{In this trivial case, a reader might ask a question: What symmetry is spontaneously broken? 
As mentioned above, the interaction induces a long-range order of a charge density wave with the spontaneous breakdown 
of the particle-hole symmetry. 
Since the charge density operator corresponds to the mass operator of {the fermions}, 
the long-range order of the charge density wave induces a non-trivial change for the expectation value of 
the mass operator of {the fermions}. 
Although the interaction becomes trivial in the continuum limit, the change remains non-zero in the limit by the above assumption. 
Namely, it appears as the mass generation for the {fermions}. 
Therefore, the chiral symmetry is broken in the continuum limit,  
although the particle-hole symmetry is no longer visible in the limit.
}

We should also remark some numerical studies for both of quantum chromodynamics and condensed matter physics as follows. 

As is well known, there have been done many numerical simulations of lattice QCD even restricted to the present subject. 
Among them, let us focus on a recent study \cite{Aoki} whose results are closely related to ours. 
They studied two-flavor QCD, and concluded that
the signal of chiral susceptibility is dominated by the axial U(1) breaking effect rather than that of
${\rm SU(2)}_L\times {\rm SU(2)}_R$ at high temperatures. 

On the other hand, as mentioned above, Mott transition occurs due to the strong Coulomb interactions between fermions. 
Clearly, the Coulomb interactions are induced by U(1) gauge fields. 
Notice that the particle-hole symmetry is a discrete $\mathbb{Z}_2$ symmetry and 
that U(1) symmetry is not broken in the phase transition. Therefore, there appears no Nambu--Goldstone mode. 
These observations imply that U(1) itself or {the} U(1) part of a generic symmetry plays an important role 
for mass generation phenomena without continuous symmetry breaking. 
 
An early numerical study about Mott transition was given by \cite{GSST}.  
For recent fermion sign free Monte Carlo studies, see \cite{LJY1,LJY2}. They showed the positivity of the weight for 
Monte Carlo simulation by using a Majorana representation of the complex fermions. 
The idea is very similar to the reflection positivity for Majorana fermions \cite{JP}. 
In particular, the second paper \cite{LJY2} treated a charge density wave transition on the $\pi$-flux square lattice. 
As mentioned above, the hopping Hamiltonian is equivalent to that of the staggered fermions. 
Therefore, the corresponding system has reflection positivity. 
{From} these numerical results, one can expect that Mott transition occurs for a much wider class of fermion systems 
than the staggered fermions beyond the restriction of reflection positivity.   

The present paper is organized as follows: In Sec.~\ref{ModelResult}, we describe our model, 
and present the precise statements of our main theorems. Our strategy to prove the main theorems is given 
in Sec.~\ref{RPGD}, i.e., we use reflection positivity for fermions and prove a Gaussian domination bound.  
In Sec.~\ref{LROnonzero}, the existence of the long-range order is proved at non-zero temperatures in 
the dimensions $\nu\ge 3$. The case at zero temperature in the dimensions $\nu\ge 2$ is proved in Sec.~\ref{LROzero}. 
A short review for staggered fermions is presented in Appendix~\ref{staggeredfermion}. 
In Appendix~\ref{Sec:chiralint}, we show that the interaction Hamiltonian of our model has a discrete chiral symmetry. 
A formal infinite-volume and continuum limit is discussed in Appendix~\ref{Dlimit}.

\section{Models and Results}
\label{ModelResult}
In this section, we describe the Hamiltonian which we will consider, 
and present the precise statements of our main results. 

Let us consider the Kogut--Susskind Hamiltonian, which takes a form of the Nambu--Jona-Lasinio (NJL) type.
Let $\Lambda := [-L+1, L]^\nu$ be the $\nu-$dimensional lattice with a positive integer $L$ and $\nu \ge 2$.
We impose the periodic boundary condition for the lattice, i.e., 
we will identify $(L+1, x^{(2)}, \dots, x^{(\nu)}) = (-L+1, x^{(2)}, \dots, x^{(\nu)})$, etc.
For each $x = (x^{(1)}, \dots, x^{(\nu)}) \in \Lambda$ we introduce a fermionic operator $\psi(x)$ which satisfies 
the anti-commutation relations, 
\begin{equation}
\label{anticommupsi}
\{\psi(x), \psi^\dagger(y)\} = \delta_{x, y}, \quad \{\psi(x), \psi(y)\}=0, 
\end{equation}
for $x, y \in \Lambda$. The staggered fermion Hamiltonian which we will consider is given by 
\begin{align}
\label{HamLambdam}
H^{(\Lambda)}(m)& := i\kappa \sum_{x \in \Lambda} \sum_{\mu=1}^\nu (-1)^{\theta_\mu(x)}
[\psi^\dagger(x) \psi(x+e_\mu) - \psi^\dagger(x+e_\mu)\psi(x)] \nonumber\\
&\quad+ m \sum_{x \in \Lambda}(-1)^{\sum_{\mu=1}^\nu x^{(\mu)}} \rho(x) 
+g\sum_{x \in \Lambda} \sum_{\mu=1}^\nu \rho(x) \rho(x+e_\mu),
\end{align}
where $\rho(x)$ is the charge density operator at the site $x$ which is defined by the deviation from the half filling, i.e., 
\begin{equation}
\label{defrho}
\rho(x) := \psi^\dagger(x)\psi(x) -\frac{1}{2}, 
\end{equation}
and the staggered {phases} of the hopping Hamiltonian are determined by 
\begin{equation*}
\theta_1(x):=
\begin{cases}
0 & \mbox{for \ } x^{(1)}\ne L;\\
1 & \mbox{for \ } x^{(1)}=L,
\end{cases}
\end{equation*}
and for $\mu=2,3,\ldots,\nu$, 
\begin{equation*}
\theta_\mu(x):=
\begin{cases}
x^{(1)}+\cdots+x^{(\mu-1)} & \mbox{for \ } x^{(\mu)}\ne L;\\
x^{(1)}+\cdots+x^{(\mu-1)}+1 & \mbox{for \ } x^{(\mu)}= L;\\
\end{cases}
\end{equation*}
the three parameters satisfy $\kappa\in \mathbb{R}$, $m\in \mathbb{R}$ and $g\ge 0$, 
and $e_\mu$ is the unit vector whose $\mu$-th component is $1$. 
Here, the additional $+1$ in the case of $x^{(\mu)}=L$ in $\theta_\mu(x)$ implies the anti-periodic boundary condition. 
For the reflection positivity of the system, we need this condition. 

Clearly, the third sum in the right-hand side of (\ref{HamLambdam}) is the interaction term, 
which is the repulsive Coulomb interaction because of the assumption $g\ge 0$. 
The second sum is the mass term of the Hamiltonian $H^{(\Lambda)}(m)$. 
As shown in Appendices~\ref{staggeredfermion} and \ref{Sec:chiralint}, in a formal continuum limit, 
this term yields the usual mass terms, $\hat{\Psi}_{\rm u}^\dagger\gamma_0\hat{\Psi}_{\rm u}$ 
and $\hat{\Psi}_{\rm d}^\dagger\gamma_0\hat{\Psi}_{\rm d}$, for the four-component Dirac spinors, 
$\hat{\Psi}_{\rm u}$ and $\hat{\Psi}_{\rm d}$, of 
the u- and d-quarks in the Hamiltonian density in the spatial dimensions $\nu=3$. Here, $\gamma_0$ is the usual Dirac's gamma matrix.

We define an order parameter, 
\begin{equation}
\label{O}
O^{(\Lambda)}:=\sum_{x \in \Lambda}(-1)^{\sum_{\mu=1}^\nu x^{(\mu)}} \rho(x), 
\end{equation}
which measures the long-range order of the charge-density wave as mentioned in Introduction. 
Clearly, this is also nothing but the mass operator in the mass term of the Hamiltonian $H^{(\Lambda)}(m)$. 
As we will see below, the mass parameter $m$ plays the role of the infinitesimally weak external field 
to induce a spontaneous magnetization 
for the order parameter $O^{(\Lambda)}$ in the infinite-volume limit.

Next, in order to see that the Hamiltonian $H^{(\Lambda)}(m)$ of (\ref{HamLambdam}) with $m=0$ has a particle-hole symmetry, 
we introduce \cite{FILS} 
\begin{equation}
\label{uxsigma}
u(x):=\left[\prod_{\substack{y\in\Lambda,\;  :\; y\ne x\; }}
(-1)^{n(y)}\right][\psi^\dagger(x)+\psi(x)], 
\end{equation}
where we have written  
\begin{equation*}
n(x):=\psi^\dagger(x)\psi(x)
\end{equation*}
for $x\in\Lambda$. Then, one has 
\begin{equation*}
[u(x)]^\dagger \psi(y)u(x)=
\begin{cases}
\psi^\dagger(x), & \mbox{for}\; y=x ;\\
\psi(y), & \mbox{otherwise}.
\end{cases}
\end{equation*}
By using these operators, we define the particle-hole transformation by \cite{FILS}  
\begin{equation*}
\label{defU2}
U_{\rm PH}^{(\Lambda)}:=\prod_{x\in\Lambda}u(x).
\end{equation*}
Immediately, 
\begin{equation*}
\label{UPH}
(U_{\rm PH}^{(\Lambda)})^\dagger \psi(x)U_{\rm PH}^{(\Lambda)}=\psi^\dagger(x)
\end{equation*}
for all $x\in \Lambda$. 
By using this, the anticommutation relations (\ref{anticommupsi}) and the definition (\ref{defrho}) of the operator $\rho(x)$, 
one has 
\begin{eqnarray*}
(U_{\rm PH}^{(\Lambda)})^\dagger[\psi^\dagger(x)\psi(x+e_\mu)-\psi^\dagger(x+e_\mu)\psi(x)]U_{\rm PH}^{(\Lambda)}
&=&\psi(x)\psi^\dagger(x+e_\mu)-\psi(x+e_\mu)\psi^\dagger(x)\nonumber\\
&=&\psi^\dagger(x)\psi(x+e_\mu)-\psi^\dagger(x+e_\mu)\psi(x)\nonumber
\end{eqnarray*}
and 
\begin{equation}
\label{rhotrans}
(U_{\rm PH}^{(\Lambda)})^\dagger\rho(x)U_{\rm PH}^{(\Lambda)}=-\rho(x).
\end{equation}
When the mass parameter $m$ of the Hamiltonian $H^{(\Lambda)}(m)$ of (\ref{HamLambdam}) is equal to zero, 
these imply the invariance of the Hamiltonian $H^{(\Lambda)}(0)$ under the particle-hole transformation, i.e., 
\begin{equation}
\label{Haminv}
(U_{\rm PH}^{(\Lambda)})^\dagger H^{(\Lambda)}(0)U_{\rm PH}^{(\Lambda)}=H^{(\Lambda)}(0).
\end{equation}

We write 
\begin{equation}
\label{expectationbetam}
\langle\cdots\rangle_{\beta,m}^{(\Lambda)}:=\frac{1}{Z_{\beta,m}^{(\Lambda)}}
\mathrm{tr} (\cdots)e^{-\beta H^{(\Lambda)}(m)}
\end{equation}
for the expectation value, where $\beta$ is the inverse temperature and 
the partition function $Z_{\beta,m}^{(\Lambda)}=\mathrm{tr} e^{-\beta H^{(\Lambda)}(m)}$.  
{From} (\ref{rhotrans}) and (\ref{Haminv}), one has 
\begin{equation}
\label{rhoexpec}
\langle\rho(x)\rangle_{\beta,m=0}^{(\Lambda)}=0.
\end{equation}
for the parameter $m=0$. Our goal is to prove that this symmetry is spontaneously broken, i.e., 
the order parameter $O^{(\Lambda)}$ 
which is defined in (\ref{O}) shows a non-trivial expectation value {with respect to a certain state}   
in the infinite-volume limit $\Lambda \nearrow \mathbb Z^\nu$.  

For this purpose, we first define the long-range order parameter by
\[
m_\mathrm{LRO}^{(\Lambda)}(\beta) :=  \vert\Lambda\vert^{-1}\sqrt{\langle [O^{(\Lambda)}]^2  \rangle_{\beta,  0}^{(\Lambda)}}.
\]
Additionally,
\[
m_\mathrm{LRO}(\beta) := \lim_{\Lambda \nearrow \mathbb Z^\nu}m_\mathrm{LRO}^{(\Lambda)}(\beta).
\]
When the quantity $m_\mathrm{LRO}(\beta)$ is non-vanishing, there exists the long-range order 
in the infinite-volume.
A standard way to show the spontaneous symmetry breaking is to prove the existence of long-range order 
in the system \cite{Griffiths}. We write 
\begin{equation}
\langle\cdots\rangle_{\beta,0}:=\lim_{\Lambda \nearrow \mathbb Z^\nu}\langle \cdots\rangle_{\beta,  0}^{(\Lambda)}.
\end{equation}
Here, if necessary, we choose a sequence of finite lattices $\Lambda$ so that the state is well-defined in the infinite-volume limit.
Clearly, from (\ref{rhoexpec}) and the definition (\ref{O}) of the order parameter $O^{(\Lambda)}$, 
this state does not show a non-vanishing magnetization for the present order parameter.  
One of the physically natural states is given by 
\begin{equation}
\langle\cdots\rangle_{\beta}:=\lim_{m\searrow 0}\lim_{\Lambda \nearrow \mathbb Z^\nu}\langle \cdots\rangle_{\beta,  m}^{(\Lambda)}.
\end{equation}
Namely, the infinitesimally weak external field is applied to the system so that the corresponding magnetization is induced.

Once one shows the existence of long-range order, the Koma--Tasaki theorem~\cite{KT} yields 
that {the physically natural state exhibits the spontaneous symmetry breaking.
In fact, the staggered magnetization, $m_s$, is non-zero in the following sense:
\[
m_s
:=
\lim_{m\searrow 0}\lim_{\Lambda\nearrow\mathbb{Z}^\nu}\frac{1}{\vert\Lambda\vert}\left\langle O^{(\Lambda)}\right\rangle_{\beta, m}^{(\Lambda)}
=
\lim_{m\searrow 0}\lim_{\Lambda\nearrow\mathbb{Z}^\nu}\frac{1}{\vert\Lambda\vert}\sum_{x \in \Lambda}\left\langle (-1)^{\sum_{\mu=1}^\nu x^{(\mu)}}\rho(x)\right\rangle_{\beta, m}^{(\Lambda)} 
 \neq 0.
\]
See~\cite{KT} for more details.
}

Our main results are as follows: For the non-zero temperatures $\beta<\infty$, we obtain 

\begin{theorem}
\label{thm.chiral}
For the dimensions $\nu \ge 3$, there exist a positive number $\alpha_0$ small enough and a positive number $\beta_0$ such that 
$m_\mathrm{LRO}(\beta)>0$ for $\vert\kappa\vert/g \le \alpha_0$ and $\beta \ge \beta_0$.
\end{theorem}

In the ground states, we can prove the existence of long-range order for the dimensions $\nu \ge 2$. 
We write 
\begin{equation}
\omega_0^{(\Lambda)}(\cdots):=\lim_{\beta \nearrow \infty}\langle\cdots\rangle_{\beta,0}^{(\Lambda)}. 
\end{equation}
Then, we obtain 

\begin{theorem}
\label{thm.GSLRO}
For the dimensions $\nu \ge 2$, there exists a positive number $\alpha_0$ small enough such that 
\[
\lim_{\Lambda \nearrow \mathbb{Z}^\nu}\frac{1}{\vert\Lambda\vert^2}\omega_0^{(\Lambda)}([O^{(\Lambda)}]^2)>0
\]
for $\vert\kappa\vert/g \le \alpha_0$.
\end{theorem}

In order to obtain a symmetry breaking vacuum (ground state), we write 
\begin{equation}
\omega_{0,m}^{(\Lambda)}(\cdots):=\lim_{\beta\nearrow\infty}\langle\cdots\rangle_{\beta,m}^{(\Lambda)},
\end{equation}
where we recall that the expectation value in the right-hand side is given by (\ref{expectationbetam}). 
Then, the spontaneous magnetization for the order parameter $O^{(\Lambda)}$ is given by 
\begin{equation}
\lim_{m\searrow 0}\lim_{\Lambda\nearrow\mathbb{Z}^\nu}\frac{1}{\vert\Lambda\vert}\omega_{0,m}^{(\Lambda)}(O^{(\Lambda})=
\lim_{m\searrow 0}\lim_{\Lambda\nearrow\mathbb{Z}^\nu}\frac{1}{\vert\Lambda\vert}\sum_{x\in\Lambda}(-1)^{\sum_{\mu=1}^\nu x^{(\mu)}}
\omega_{0,m}^{(\Lambda)}(\psi^\dagger(x)\psi(x)). 
\end{equation}
Since the existence of the long-range order implies a non-vanishing spontaneous magnetization \cite{KT}, 
this right-hand side is non-vanishing. Thus, the expectation value of the mass term of the Hamiltonian 
is non-vanishing in the symmetry breaking vacuum (ground state). 
Naively, one can expect that the expectation value of the charge density operator $\psi^\dagger(x)\psi(x)$ 
takes the value $1/2$. But this result implies that there appears a non-trivial deviation from $1/2$. 
This is nothing but the emergence of a charge density wave order. 
In passing, the particle-hole symmetry is locally broken.

\begin{remark}
We compare the results of Salmhofer and Seiler~\cite{SS1} with ours. To begin with, let us recall 
the models and the results.
In~\cite{SS1}, they dealt with staggered fermions coupled to U(N) gauge fields on the hypercubic lattice 
in the Euclidean formalism with the Lagrangian density. In the present paper, we use the Hamiltonian formalism 
for the staggered fermions with the effective four-fermion interaction without gauge fields.
Therefore, our model corresponds to their U(1) case. Namely, 
among their settings~\cite{SS1}, we have treated only the corresponding U(1) case. 
In~\cite{SS1}, they took the strong coupling limit for the coupling constant of the self-interaction of the gauge fields. 
Because of the strong coupling limit, the integral over the gauge fields 
can be done, and the resulting action consists of the nearest-neighbor four-fermion interactions and the mass term 
without usual hopping terms. Both of the interaction and mass terms are equivalent to ours in the U(1) case, 
although their effective model loses hopping terms.

In~\cite{SS1}, they proved the chiral symmetry breaking in the sense of the non-vanishing of the expectation value of 
the mass observable when the space-time dimension is greater than or equal to four. As mentioned above, this mass observable is 
equivalent to our order parameter $O^{(\Lambda)}$ of (\ref{O}).

Apart from the models themselves and their results, 
other differences between~\cite{SS1} and ours are summarized as follows: 
\begin{enumerate}
\item In the U(1) case of \cite{SS1}, they treated the lattice gauge theory 
with $N_f=2^{\lfloor d/2\rfloor}$-flavors of Dirac fermions 
in the $d=\nu+1$ space-time dimensions: $(d, N_f)= (4, 4), (5, 4), (6, 8), (7, 8), \dots$. 
Here, ${\lfloor b \rfloor}:=\max\{\ell\in\mathbb{Z}\colon \ell\le b \}$ for a real number $b$. 
In our setting, the pairs of the number of flavors of Dirac fermions, which is given by $N_f=2^{\lfloor \nu/2\rfloor}$, 
and the space-time dimensions are $(d, N_f)= (3, 2), (4, 2), (5, 4), (6, 4), \dots$.
For instance, in the four space-time dimensions, there appear four flavors of the Dirac fermions 
because the time is also discretized in the Euclidean formulation which relies on Grassmann algebra.  
On the other hand, in our setting with three spatial dimensions, there are two-fermion flavors 
(see also Appendix~\ref{staggeredfermion}).
In other words, the discretization of time leads to additional fermions due to the doubling problem.
These differences are arise due to the difference between the Lagrangian and the Hamiltonian formalisms.
\item From a technical point of view, there is a considerable difference between the two models about the reflection positivity.
In~\cite{SS1}, the model can be reduced to a simple spin system with only nearest-neighbor spin-spin interactions, 
such as monomer-dimer systems, by using the special property of the U(N) integral. 
Thus, their method is based on the reflection positivity for spin systems.
In contrast to~\cite{SS1}, we need the reflection positivity \cite{JP} for fermion systems. 
\item Besides, we emphasize that the main difficulty that we have to overcome in our setting is to establish reflection positivity 
for the model including the hopping term.
Such a hopping term, which disappears in the resulting effective model of~\cite{SS1}, 
is an annoyance for proving reflection positivity even in spin systems. 
\end{enumerate}
\end{remark}

\section{Reflection Positivity and Gaussian Domination}
\label{RPGD}

The aim of this section is to prove the Gaussian domination bound which is given in Proposition~\ref{GaussianD} below, 
by using the reflection positivity of the present system. The bound leads to an estimate for the long-range order. 

Let $\Omega \subset \Lambda$ be a subset and $\mathcal{A}(\Omega)$ be the algebra generated by 
$\psi(x)$ and $\psi^\dagger(y)$ for  $x, y \in \Omega$.
Since our $\Lambda$ is symmetric with respect to a hyperplane, 
there are a natural decomposition $\Lambda = \Lambda_- \cup \Lambda_+$ with $\Lambda_- \cap \Lambda_+ = \emptyset$ 
and a reflection map $r: \Lambda_{\pm} \to \Lambda_{\mp}$ satisfying $r(\Lambda_\pm) = \Lambda_\mp$.
We write $\mathcal{A} = \mathcal{A}(\Lambda)$ and $\mathcal{A}_\pm = \mathcal{A}(\Lambda_\pm)$. 
The reflection has an anti-linear representation $\vartheta: \mathcal{A}_\pm \to \mathcal{A}_\mp$ requiring 
\begin{align*}
&\vartheta(\psi(x)) = \psi(\vartheta(x)), \quad \vartheta(\psi^\dagger(x))= \psi^\dagger(\vartheta(x)),\\
 &\vartheta(AB) = \vartheta(A)\vartheta(B), \quad \vartheta(A)^\dagger = \vartheta(A^\dagger) \quad \text{for } A, B \in \mathcal{A}.
 \end{align*}

For $x \in \Lambda $ we introduce Majorana fermion operators $\xi(x), \eta(x)$ by
\begin{equation}
\label{Majoranapsi}
\xi(x) := \psi^\dagger(x) + \psi(x), \quad \eta(x) := i(\psi^\dagger(x) - \psi(x)),
\end{equation}
or equivalently, 
\begin{equation}
\label{psiMajorana}
\psi(x) = \frac{1}{2}(\xi(x) + i\eta(x)), \quad \psi^\dagger(x) = \frac{1}{2}(\xi(x) - i\eta(x)).
\end{equation}
These satisfy $\xi^\dagger(x) =\xi(x)$, $\eta^\dagger(x) = \eta(x)$, and the anti-commutation relations
\begin{align*}
\{\xi(x), \xi(y)\} &= 2\delta_{x, y},  \quad \{{\eta}(x), {\eta}(y)\} =2\delta_{x, y}, \\
\{\xi(x), \eta(y)\} &=0.
\end{align*}
Moreover, we have
\begin{equation}
\label{varthetaMajorana}
(\vartheta\xi)(x) = \xi(\vartheta(x)), \quad (\vartheta\eta)(x) = - \eta(\vartheta(x)).
\end{equation}

For later convenience, we write the hopping term of our Hamiltonian (\ref{HamLambdam}) as
\[
H_K^{(\Lambda)} := \sum_{\mu=1}^\nu H_{K, \mu}^{(\Lambda)},
\]
with
\begin{equation}
\label{HKmu}
H_{K, \mu}^{(\Lambda)} := i\kappa \sum_{x \in \Lambda}(-1)^{\theta_\mu(x)}
[\psi^\dagger(x) \psi(x+e_\mu) - \psi^\dagger(x+e_\mu)\psi(x)].
\end{equation}
Moreover, we note that the interaction is {equal} to
\begin{align*}
g\sum_{x \in \Lambda} \sum_{\mu=1}^\nu \rho(x) \rho(x+e_\mu)
&= \frac{g}{2}\sum_{x \in \Lambda} \sum_{\mu=1}^\nu [\rho(x) + \rho(x+e_\mu)]^2 -\frac{g\nu\vert\Lambda\vert}{4},
\end{align*}
where we have used $\rho(x) \rho(y) = \rho(y) \rho(x)$ and $\rho(x)^2 = 1/4$.
Let $h^{(\mu)} \colon \Lambda \to \mathbb{R}$ be real-valued functions.
Following the idea of~\cite{FSS,DLS}, we define a new interaction as
\begin{equation}
\label{Hinth}
H_\text{int}^{(\Lambda)}(h) := \sum_{\mu=1}^\nu H_{\text{int,}\, \mu}^{(\Lambda)}(h) -\frac{g\nu\vert\Lambda\vert}{4}
\end{equation}
with
\begin{equation}
\label{Hintmuh}
H_{\text{int}, \, \mu}^{(\Lambda)}(h) 
= \frac{g}{2}\sum_{x \in \Lambda} [\rho(x) + \rho(x+e_\mu) +(-1)^{\sum_{i=1}^\nu x^{(i)} } h^{(\mu)}(x)]^2.
\end{equation}
When $h^{(\mu)}=0$ for all $\mu$, the interaction equals the original one.
We will also consider the corresponding total Hamiltonian, 
\[
H^{(\Lambda)}(m, h) := H_K^{(\Lambda)} + H_\text{int}^{(\Lambda)}(h)+mO^{(\Lambda)}.
\]
Clearly, when $h^{(\mu)}=0$ for all $\mu$, we have the original total Hamiltonian, i.e., $H^{(\Lambda)}(m, 0)=H^{(\Lambda)}(m)$.

First, we divide our finite lattice $\Lambda$ by the $x^{(1)} = 1/2$ hyperplane into 
$\Lambda_- :=\{x \in \Lambda \colon -L+1 \le x^{(1)} \le 0\}$ and 
$\Lambda_+ := \{x  \in \Lambda \colon 1 \le x^{(1)} \le L\}$.
As in~\cite{Koma4, FILS} we introduce some unitary transformations:
 For $j = 2, \dots, \nu$,
\[
U_{1, j} := \prod_{\substack{x\in \Lambda,\\ x^{(j)}= \text{ even}}} e^{(\frac{i\pi}{2})n(x)},
\]
with $n(x)=\psi^\dagger(x)\psi(x)$,
and
\begin{equation}
\label{U1}
U_1 := \prod_{j=2}^\nu U_{1, j}.
\end{equation}
One can see
\begin{equation*}
U_{1, j}^\dagger \psi(x)U_{1, j} = 
\begin{cases}
i\psi(x) & \mbox{for \ } x^{(j)} \text{ is even};\\
\psi(x) & \mbox{for \ } x^{(j)} \text{ is odd}.
\end{cases}
\end{equation*}
Next, consider the transformation for the Hamiltonian $H_{K, \mu}^{(\Lambda)}$ of (\ref{HKmu}). 
We have  
\[
U_{1, j}^\dagger H_{K, i }^{(\Lambda)}U_{1, j}  = H_{K, i }^{(\Lambda)}\quad\mbox{for \ }i \neq j,
\]
and 
\begin{align*}
U_{1, j}^\dagger H_{K, j}^{(\Lambda)}U_{1, j} 
&= \kappa \sum_{\substack{x \in \Lambda \\ x^{(j)} \neq L}} (-1)^{\sum_{i=1}^j x^{(i)}} [\psi^\dagger(x) \psi(x+e_j) + \text{h.c.}] \\
&\quad -\kappa \sum_{\substack{x \in \Lambda \\ x^{(j)} = L}} (-1)^{\sum_{i=1}^j x^{(i)}} [\psi^\dagger(x) \psi(x+e_j) + \text{h.c.}].
\end{align*}
Hence we have 
\begin{equation}
\label{transU1HK1}
U_1^\dagger H_{K, 1}^{(\Lambda)} U_1 = H_{K, 1}^{(\Lambda)}, 
\end{equation} 
and for $j=2, \dots, \nu$
\begin{equation}
\begin{split}
\label{eq:hop1}
U_{1}^\dagger H_{K, j}^{(\Lambda)}U_{1} 
&= \kappa \sum_{\substack{x \in \Lambda \\ x^{(j)} \neq L}} (-1)^{\sum_{i=1}^j x^{(i)}} [\psi^\dagger(x) \psi(x+e_j) + \text{h.c.}] \\
&\quad -\kappa \sum_{\substack{x \in \Lambda \\ x^{(j)} = L}} (-1)^{\sum_{i=1}^j x^{(i)}} [\psi^\dagger(x) \psi(x+e_j) + \text{h.c.}].
\end{split}
\end{equation}

In order to introduce the second unitary transformation, 
let $\Lambda_{\text{odd}} :=\{x \in \Lambda \colon x^{(1)} + \cdots + x^{(\nu)} = \text{odd}\}$.  
Then, {a} unitary transformation is defined by 
\begin{equation}
\label{Uodd}
U_\text{odd} := \prod_{x \in \Lambda_\text{odd}} u(x),
\end{equation}
where $u(x)$ is defined in (\ref{uxsigma}). It is easy to see that
\begin{equation}
\label{Uoddtranspsi}
U_\text{odd}^\dagger \psi(x) U_\text{odd} 
=
\begin{cases}
\psi^\dagger(x) & \mbox{for } x \in \Lambda_\text{odd};\\
\psi(x) & \mbox{otherwise} .
\end{cases}
\end{equation}

Now we put $\tilde U_1 := U_1 U_\text{odd}$ by using the above two unitary transformations $U_1$ of (\ref{U1}) 
and $U_\text{odd}$. Then by (\ref{eq:hop1}) we see that for $j=2, \dots, \nu-1$
\begin{equation}
\begin{split}
\tilde{H}_{K, j}^{(\Lambda)} := \tilde U_{1}^\dagger H_{K, j}^{(\Lambda)} \tilde U_{1} 
&=
\kappa \sum_{\substack{x \in \Lambda \\ x^{(j)} \neq L}} (-1)^{\sum_{i=j+1}^\nu x^{(i)}} 
[\psi^\dagger(x) \psi^\dagger(x+e_j) + \text{h.c.}] \\
&\quad -\kappa \sum_{\substack{x \in \Lambda \\ x^{(j)} = L}}  
(-1)^{\sum_{i=j+1}^\nu x^{(i)}} [\psi^\dagger(x) \psi^\dagger(x+e_j) + \text{h.c.}],
\end{split}
\end{equation}
and
\begin{equation}
\begin{split}
\tilde{H}_{K, \nu}^{(\Lambda)} := \tilde U_{1}^\dagger H_{K, \nu}^{(\Lambda)} \tilde U_{1} 
&=
\kappa \sum_{\substack{x \in \Lambda \\ x^{(\nu)} \neq L}} [\psi^\dagger(x) \psi^\dagger(x+e_\nu) + \text{h.c.}] \\
&\quad -\kappa \sum_{\substack{x \in \Lambda \\ x^{(\nu)} = L}} [\psi^\dagger(x) \psi^\dagger(x+e_\nu) + \text{h.c.}].
\end{split}
\end{equation}
These Hamiltonian are expressible as
\[
\tilde{H}_{K, j}^{(\Lambda)} = \tilde{H}_{K, j}^+ + \tilde{H}_{K, j}^-,
\]
for $j=2, \dots, \nu$, where $\tilde{H}_{K, j}^\pm \in \mathcal{A}_\pm$ satisfy $\vartheta(\tilde{H}_{K, j}^\pm) =\tilde{H}_{K, j}^\mp$.
For $H_{K, 1}^{(\Lambda)}$ of (\ref{HKmu}), we use the relation
\[
\psi^\dagger(x) \psi(y) - \psi^\dagger(y) \psi(x) = \frac{1}{2}[\xi(x) \xi(y) +\eta(x) \eta(y)],
\]
which can be derived from the expressions (\ref{psiMajorana}). 
One has $U^\dagger_\text{odd} \xi(x) U_\text{odd} = \xi(x)$ and
\begin{equation*}
U_\text{odd}^\dagger \eta(x) U_\text{odd} =
\begin{cases}
-\eta(x) & \mbox{for } x \in \Lambda_\text{odd};\\
\eta(x) & \mbox{otherwise} 
\end{cases}
\end{equation*}
{from} (\ref{Majoranapsi}) and (\ref{Uoddtranspsi}). 
By combining these with (\ref{transU1HK1}) and $\tilde U_1 := U_1 U_\text{odd}$, we have
\begin{equation}
\begin{split}
\tilde{H}_{K, 1}^{(\Lambda)} := \tilde U_{1}^\dagger H_{K, 1}^{(\Lambda)} \tilde U_{1} 
&=
\frac{i\kappa}{2} \sum_{\substack{x \in \Lambda \\ x^{(1)} \neq L}} [\xi(x) \xi(x+e_1) - \eta(x) \eta(x+e_1) ]\\
&\quad -\frac{i\kappa}{2} \sum_{\substack{x \in \Lambda \\ x^{(1)} = L}} [\xi(x) \xi(x+e_1) - \eta(x) \eta(x+e_1) ].
\end{split}
\end{equation}
Hence the Hamiltonian has the form
\[
\tilde{H}_{K, 1}^{(\Lambda)} = \tilde{H}_{K, 1}^{+} + \tilde{H}_{K, 1}^{-} +\tilde{H}_{K}^{0},  
\]
where
\[
\tilde{H}_{K, 1}^{\pm} := 
 \frac{i\kappa}{2} \sum_{\substack{x, x+e_1 \in \Lambda_\pm}} [\xi(x) \xi(x+e_1) - \eta(x) \eta(x+e_1) ] ,
\]
and
\[
\tilde{H}_{K}^{0}
:= 
 \frac{i\kappa}{2}  \sum_{\substack{x \in \Lambda \\ x^{(1)} = 0, -L+1}}[\xi(x) (\vartheta\xi)(x) + \eta(x) (\vartheta\eta)(x)].
\]
Here, we have used the relations (\ref{varthetaMajorana}) for the Majorana fermions. 
Then $\tilde{H}_{K, 1}^{\pm} \in \mathcal{A}_\pm$ and $\vartheta(\tilde{H}_{K, 1}^{\pm})=\tilde{H}_{K, 1}^{\mp}$ follows.
{From} these observations, we have
\[
\tilde{H}_{K} := \tilde U_1^\dagger \tilde{H}_{K}^{(\Lambda)} \tilde U_1
= \sum_{j=1}^\nu \tilde{H}_{K, j}^{+} + \sum_{j=1}^\nu \tilde{H}_{K, j}^{-} + \tilde{H}_{K}^0 =: \tilde{H}_{K} ^+ + \tilde{H}_{K}^- +\tilde{H}_{K}^0.
\]

For the interaction $H_\text{int}^{(\Lambda)}(h)$ of (\ref{Hinth}), we have that $U_1^\dagger \rho(x) U_1 = \rho(x)$ and
\begin{equation*}
U_\text{odd}^\dagger \rho(x) U_\text{odd} =
\begin{cases}
-\rho(x) & \mbox{for } x \in \Lambda_\text{odd};\\
\rho(x) & \mbox{otherwise}
\end{cases}
\end{equation*}
in the same way. Hence
\[
\tilde H_\text{int}^{(\Lambda)}(h) := \tilde U_1^\dagger H_\text{int}^{(\Lambda)}(h) \tilde U_1 = 
\frac{g}{2}\sum_{\mu =1}^\nu \sum_{x \in \Lambda} [\rho(x) - \rho(x+e_\mu) +h^{(\mu)}(x)]^2 -\frac{g\nu\vert\Lambda\vert}{4}.
\]
Then the interaction Hamiltonian can be written as
\[
\tilde H_\text{int}^{(\Lambda)}(h) = \tilde H_\text{int}^{+}(h) + \tilde H_\text{int}^{-}(h) + \tilde H_\text{int}^{0}(h),
\]
where $\tilde H_\text{int}^{\pm}(h) \in \mathcal{A}_\pm$ (but $\tilde H_\text{int}^{\pm}(h) \neq 
\vartheta(\tilde H_\text{int}^{\mp}(h))$ when $h \neq 0$) and
\[
\tilde H_\text{int}^{0}(h)
:=
\frac{g}{2}\sum_{\substack{x \in \Lambda \\ x^{(1)} = 0, L}}[\rho(x) - \rho(x+e_1) +h^{(1)}(x)]^2.
\]
Similarly, we see
\[
\tilde O^{(\Lambda)} := \tilde U_1^\dagger O^{(\Lambda)} \tilde U_1 = 
\sum_{x \in \Lambda} \rho(x),
\]
and
\[
\tilde H_\text{SB}(m) := m\tilde O^{(\Lambda)} 
= \tilde H_\text{SB}^{+}(m) +\tilde H_\text{SB}^{-}(m),
\]
with $H_\text{SB}^{\pm}(m) \in \mathcal{A}_\pm$ satisfying $\vartheta(H_\text{SB}^{\pm}(m))= H_\text{SB}^{\mp}(m)$.

We conclude that
\begin{align}
\label{exprestildeHLambdamh}
\tilde H^{(\Lambda)}(m, h) &:= \tilde U_1^\dagger  H^{(\Lambda)}(m, h) \tilde U_1 \nonumber\\
&= \tilde H^{+}(m, h) + \tilde H^{-}(m, h)  + \tilde H^{0}(m, h),
\end{align}
where $\tilde H^{\pm}(m, h) \in \mathcal{A}_\pm$ and
\[
\tilde H^{0}(m, h) := \tilde{H}_{K}^{0} + \tilde{H}_\text{int}^{0}(h).
\]

We will prove the following proposition:
\begin{proposition}[Gaussian domination]
\label{GaussianD}
For any real-valued $h = (h^{(1)}, \dots, h^{(\mu)})$, it follows that
\begin{align}
\label{GaussianDbund}
\mathrm{tr} \exp[- \beta H^{(\Lambda)}(m, h)] 
&\le
\mathrm{tr} \exp[- \beta H^{(\Lambda)}(m, 0)].
\end{align}
\end{proposition}

\begin{proof}
Using the Lie-Trotter product formula, we have 
$$
\mathrm{tr} \exp[- \beta \tilde H^{(\Lambda)}(m, h)]= \lim_{n \to \infty}\mathrm{tr}( \alpha_n^n)
$$ 
with
\[
\alpha_n
:=
\left(1 - \frac{\beta}{n} \tilde H^{0}_{K, 1} \right)
\left[\prod_{\substack{x \in \Lambda: \\ x^{(1)} =0, L}} e^{-\frac{\beta g}{2n}(\rho(x) - \rho(x+ e_1) + h^{(1)}(x))^2}\right] 
e^{- \frac{\beta}{n} \tilde H^{-}(m, h)} e^{-\frac{\beta}{n} \tilde H^+(m, h)}
\]
For any real hermitian $D$, we can write
\[
e^{-D^2} = (4 \pi)^{-1/2} \int_\mathbb{R} dk\; e^{ikD} e^{-k^2/4}.
\]
Hence we see
\begin{align*}
e^{-\frac{\beta g}{2n}(\rho(x) - \rho(x+ e_1) + h^{(1)}(x))^2}
&=
 \int_\mathbb{R} \frac{dk}{\sqrt{4 \pi}} e^{-k^2/4} e^{ik \sqrt{\frac{\beta g}{2n}} \rho(x)} 
  e^{-ik \sqrt{\frac{\beta g}{2n}} \rho(x+ e_1)} e^{ik \sqrt{\frac{\beta g}{2n}}  h^{(1)}(x)}.
\end{align*}
Denoting
\begin{align*}
&K_-(k, m, h) := A_-(k) e^{-\beta \tilde H^- (m, h)/n}, \\
&K_+(k, m, h) := \underbrace{\vartheta(A_-(k))}_{=: A_+(k)} e^{-\beta \tilde H^+ (m, h)/n},\\
A_-(k) &:=  \prod_{\substack{x \in \Lambda, \\ x^{(1)} =0}} \exp\left[ik(x) \sqrt{\frac{\beta g}{2n}} \rho(x)\right]
\prod_{\substack{x \in \Lambda, \\ x^{(1)} =L}}   \exp\left[-ik(x) \sqrt{\frac{\beta g}{2n}} \rho(x+e_1)\right],
\end{align*}
we have
\begin{align*}
\mathrm{tr} \left(\alpha_n^n\right)
=
 \int d\mu(k_1)\cdots d\mu(k_n)&
\mathrm{tr} \left[\prod_{i=1}^n \left(1 - \frac{\beta}{n} \tilde H^{0}_{K, 1} \right) K_-(k_i, m, h) K_+(k_i, m, h)
\right] \\
& \times \prod_{\substack{x \in \Lambda, \\ x^{(1)} =0, L}}  
\exp\left[i \left\{k_1(x)+ \cdots+ k_n(x)\right\} \sqrt{\frac{\beta g}{2n}}  h^{(1)}(x)\right]
\end{align*}
with
\begin{align*}
\int d\mu(k)
&:=
\prod_{\substack{x \in \Lambda, \\ x^{(1)} =0, L}} \int_\mathbb{R} \frac{dk(x)}{\sqrt{4\pi}} e^{- k(x)^2/4}.
\end{align*}
We note that each term in the trace can be written
\begin{equation}
\begin{split}
\label{Eq.integrand}
\left(- \frac{i\beta \kappa}{2n} \right)^j&
K_-(k_1, m, h) K_+(k_1, m, h) \cdots K_-(k_{l_1}, m, h) K_+(k_{l_1}, m, h) \gamma_1 \vartheta(\gamma_1) \\
& \times K_-(k_{l_1+1}, m, h) K_+(k_{l_1+2}, m, h) \cdots K_-(k_{l_2}, m, h) K_+(k_{l_2}, m, h) \gamma_2 \vartheta(\gamma_2) \\
& \times K_-(k_{l_2+1}, m, h) K_+(k_{l_2+2}, m, h) \cdots K_-(k_{l_j}, m, h) K_+(k_{l_j}, m, h) \gamma_j \vartheta(\gamma_j)\\
&\times K_-(k_{l_j+1}, m, h) K_+(k_{l_j+2}, m, h) \cdots K_-(k_{n}, m, h) K_+(k_{n}, m, h),
\end{split}
\end{equation}
where $\gamma_l \in \{\xi(x), \eta(x)\}$, for $x \in \Lambda$ with $x^{(1)} = 0$ or $x^{(1)}= -L+1$.
Since $\tilde H^- (m, h)$ can be written as a sum of an even number of monomials in Majoranas, each $K_-(k, m, h)$ can also be written as a sum of even Majoranas.
From~\cite[Prop.~1]{JP}, the trace of (\ref{Eq.integrand}) vanishes 
unless each Majorana appears an even number of times on the half lattices.
Hence we consider only the case that $j$ is even.
Since $K_\pm(k, m, h) \in \mathcal{A}_\pm$ has even fermion parity, 
we see $[K_-, K_+] = 0$, $[K_+, \gamma_i] =0$, and $[K_-, \vartheta(\gamma_i)] =0$.
Combining these with $\vartheta(\gamma_1) \gamma_2 = - \gamma_2\vartheta(\gamma_1)$, we have that
\begin{align*}
&K_-(k_1, m, h) K_+(k_1, m, h) \cdots K_-(k_{l_1}, m, h) K_+(k_{l_1}, m, h) \gamma_1 \vartheta(\gamma_1) \\
& \times K_-(k_{l_1+1}, m, h) K_+(k_{l_1+2}, m, h) \cdots K_-(k_{l_2}, m, h) K_+(k_{l_2}, m, h) \gamma_2 \vartheta(\gamma_2) \\
&=
-K_-(k_1, m, h)  \cdots K_-(k_{l_1}, m, h) \gamma_1 K_-(k_{l_1 +1 }, m, h) \cdots K_-(k_{l_2}, m, h) \gamma_2  \\
& \times K_+(k_{1}, m, h)  \cdots K_+(k_{l_1}, m, h) \vartheta(\gamma_1) K_+(k_{l_1 +1 }, m, h) 
\cdots K_+(k_{l_2}, m, h) \vartheta(\gamma_2)\\
&=: -X_-(1)X_+(1),
\end{align*}
where $X_\pm(1) \in \mathcal{A}_\pm$.
Hence, for $j = 2m$, the term (\ref{Eq.integrand}) can be written in the form, 
\begin{align*}
&\left(\frac{\beta \kappa}{2n} \right)^{2m} X_-(1)X_+(1) \cdots X_-(m)X_+(m) \\
&= \left(\frac{\beta \kappa}{2n} \right)^{2m} X_-(1)X_-(2) \cdots X_-(m) X_+(1) X_+(2) \cdots X_+(m),
\end{align*}
where $X_\pm(i) \in \mathcal{A}_\pm, i=1,2,\ldots,m$, and we have used $[X_-(i), X_+(j)]=0$ for all $i,j$.

Consequently, we have
\begin{align}
\label{exprestralphann}
\mathrm{tr}(\alpha_n^n)
=\int d\mu(k_1) \cdots d\mu(k_n)& \sum_{j} \mathrm{tr} \left( W_-(j) W_+(j)\right) \nonumber\\
& \times \prod_{\substack{x \in \Lambda, \\ x^{(1)} =0, L}}  
\exp\left[i \left\{k_1(x)+ \cdots+ k_n(x)\right\} \sqrt{\frac{\beta g}{2n}}  h^{(1)}(x)\right],
\end{align}
with
\[
W_\pm(j) := \left(\frac{\beta \kappa}{2n} \right)^{m_j} X_\pm(1, j)X_\pm(2, j) \cdots X_\pm(m_j, j).
\]
Here $m_j \ge 0$ is an integer and $X_{\pm}(i, j) \in \mathcal{A}_\pm$ for $i = 1, \dots, m_j$. 

The following lemma is essentially the same as~\cite[Prop.~2]{JP}. 
\begin{lemma}
\label{lem.RP}
For any operator $A \in \mathcal{A}_\pm$ it holds that
\[
\mathrm{tr} (A \vartheta(A)) \ge 0.
\]
\end{lemma}

By Lemma~\ref{lem.RP}, the Cauchy--Schwarz inequality  $\vert\mathrm{tr}(A \vartheta(B))\vert^2 \le \mathrm{tr}(A \vartheta(A)) \mathrm{tr}(B \vartheta(B))$ 
holds for $A, B \in \mathcal{A}_-$. 
Combining this inequality with the Cauchy--Schwarz inequality for the sum and integrations, we can evaluate 
$\mathrm{tr}(\alpha_n^n)$ of (\ref{exprestralphann}) as follows: 
\begin{align*}
\vert\mathrm{tr}(\alpha_n^n)\vert
&\le
\int d\mu(k_1) \cdots d\mu(k_n) \sum_{j} \left\vert\mathrm{tr} \left( W_-(j) W_+(j)\right)\right\vert \\
&\le
\int d\mu(k_1) \cdots d\mu(k_n) \sum_{j} \sqrt{\mathrm{tr}\left[W_-(j) \vartheta (W_-(j))\right]}  \sqrt{\mathrm{tr}\left[W_+(j) \vartheta ( W_+(j))\right]}\\
&\le
\int d\mu(k_1) \cdots d\mu(k_n) 
\left(\sum_j \mathrm{tr}\left[W_-(j) \vartheta (W_-(j))\right] \right)^{1/2}
\left(\sum_j \mathrm{tr}\left[W_+(j) \vartheta ( W_+(j))\right] \right)^{1/2} \\
&\le
\left(\int d\mu(k_1) \cdots d\mu(k_n) \sum_{j} \mathrm{tr} \left[ W_-(j) \vartheta(W_-(j))\right]\right)^{1/2} \\
&\quad \times \left(\int d\mu(k_1) \cdots d\mu(k_n) \sum_{j} \mathrm{tr} \left[ W_+(j) \vartheta( W_+(j))\right]\right)^{1/2},
\end{align*}
where we have used $W_+(j) = \vartheta (\vartheta (W_+(j)))$ with $\vartheta (W_+(j)) \in \mathcal{A}_-$.

Undoing the above steps, we see that
\begin{equation}
\begin{split}
\label{Eq.trace}
\left\{\mathrm{tr} \exp[- \beta \tilde H^{(\Lambda)}(m, h)] \right\}^2
&\le
\mathrm{tr} \exp[- \beta ( \tilde H^{(-)}(m, h)) +  \vartheta(\tilde H^{(-)}(m, h)+ \tilde H^{(0)}(0)) ]\\
& \quad \times 
\mathrm{tr} \exp[- \beta (\vartheta(\tilde H^{(+)}(m, h))+ \tilde H^{(+)}(m, h)  + \tilde H^{(0)}(0)) ].
\end{split}
\end{equation}
For $h = (h^{(1)}, \dots, h^{(\nu)})$, we define $h_\pm = (h_\pm^{(1)}, \dots, h_\pm^{(\nu)})$ by
\begin{equation*}
h_\pm^{(1)} :=
\begin{cases}
h^{(1)}(x) & \mbox{if } x \in \Lambda_\pm, \, x^{(1)}\neq 0,  L \\
-h^{(1)}(\vartheta(x+e_1)) & \mbox{if } x \in \Lambda_\mp, \, x^{(1)} \neq 0, L \\
0  & \mbox{otherwise, \ }
\end{cases}
\end{equation*}
and for $i = 2, \dots, \nu$
\begin{equation*}
h_\pm^{(i)} :=
\begin{cases}
h^{(i)}(x) & \mbox{if } x \in \Lambda_\pm \\
h^{(i)}(\vartheta(x)) & \mbox{if } x \in \Lambda_\mp.
\end{cases}
\end{equation*}
Then, by (\ref{Eq.trace}), (\ref{exprestildeHLambdamh}) 
and $U e^{A} U^{-1} = e^{UAU^{-1}}$ for any matrix $A$ and invertible $U$, we have
\begin{equation}
\label{Eq.Gaussian0}
\left\{ \mathrm{tr} \exp[-\beta H^{(\Lambda)}(m, h)] \right\}^2
\le  \mathrm{tr} \exp[-\beta H^{(\Lambda)}(m, h_-)]   \mathrm{tr} \exp[-\beta H^{(\Lambda)}(m, h_+)]. 
\end{equation}

We claim that the inequality (\ref{Eq.Gaussian0}) also holds for reflections across any hyperplane.
The proof is essentially the same as in~\cite[Sect.~4]{Koma4}, but we repeat the argument here for completeness.
We consider a unitary operator
\[
U_{\mathrm{HA}}(j \to 1) := U_{\mathrm{HA}}(j, j-1) U_{\mathrm{HA}}(j, j-2) \cdots U_{\mathrm{HA}}(j, 1)
\]
with
\[
U_{\mathrm{HA}}(i, j)  := \prod_{\substack{x \in \Lambda\\  x^{(i)}, x^{(j)}  = \text{ odd} }} e^{i\pi n(x)}.
\]
Then we have
\[
U^\dagger_{\mathrm{HA}}(j \to 1) H_K^{(\Lambda)} U_{\mathrm{HA}}(j \to 1)
=i\kappa \sum_{x \in \Lambda} \sum_{\mu =1}^\nu (-1)^{\tilde \theta_\mu(x)}
[\psi^\dagger(x) \psi(x+e_\mu) - \psi^\dagger(x+e_\mu)\psi(x)] 
\]
where
\begin{equation*}
\tilde \theta_j(x):=
\begin{cases}
0 & \mbox{if \ } x^{(j)}\ne L;\\
1 & \mbox{if \ } x^{(j)}=L,
\end{cases}
\end{equation*}
$\tilde \theta_\mu = \theta_\mu$ for $\mu > j$, and for $\mu < j$, 
\begin{equation*}
\tilde \theta_\mu(x):=
\begin{cases}
x^{(1)}+\cdots+x^{(\mu-1)} +x^{(j)} & \mbox{if \ } x^{(\mu)}\ne L;\\
x^{(1)}+\cdots+x^{(\mu-1)}+ x^{(j)}+1 & \mbox{if \ } x^{(\mu)}= L.
\end{cases}
\end{equation*}
This shows that the hopping amplitudes are gauge equivalent in all directions.

Next, we consider
\[
U_{\mathrm{BC}, i} (L \to l) := \prod_{\substack{x \in \Lambda\\ l \le x^{(i)} \le L}} e^{i\pi n(x)}.
\]
If $l \le x^{(i)} \le L$ then
\[
U^\dagger_{\mathrm{BC}, i} (L \to l) \psi(x) U_{\mathrm{BC}, i} (L \to l) = - \psi(x).
\]
Hence this transformation changes the role of $\{(-L+1, x^{(2)}, \dots, x^{(\nu)}), (L, x^{(2)}, \dots, x^{(\nu)})  \}$ to 
$\{(l-1, x^{(2)}, \dots, x^{(\nu)}), (l, x^{(2)}, \dots, x^{(\nu)})  \}$, etc.
Thus, together with the above, all the reflections across any hyperplane are equivalent in the sense of these gauge equivalence.

We now turn to the proof of Gaussian domination.
Since $\mathrm{tr} \exp[-\beta H^{(\Lambda)}(m, h)]$ is bounded and continuous in $h$, there is at least one maximizer.
We choose a maximizer $h_0$ to be one that contains the maximal number of zeros.
We claim that $h_0^{(i)}(x) = 0$ for all $i, x$.
If not, we may assume that $h_0^{(1)}(L, x^{(2)}, \dots, x^{(\nu)}) \neq 0$ for some $(x^{(2)}, \dots, x^{(\nu)})$ by the above arguments.
By (\ref{Eq.Gaussian0}), $h_\pm$ obtained from $h_0$ are also maximizers.
This contradicts our assumption on $h_0$, because $h_\pm$ must contain strictly more zero than $h_0$ by definitions.
Hence $h_0 \equiv 0$.
This completes the proof.
\end{proof}

\section{Long-range order at non-zero temperatures}
\label{LROnonzero}

Now we give a proof of the existence of the long-range order for non-zero temperatures in the dimensions $\nu\ge 3$. 

By using the unitary transformation $U_{\rm odd}$ of (\ref{Uodd}), 
let $H_\mathrm{odd}(m, h) := U_\mathrm{odd}^\dagger H^\mathrm{(\Lambda)}(m, h) U_\mathrm{odd}$.
Then, by the Gaussian domination bound (\ref{GaussianDbund}), we have
\begin{equation}
\label{Eq.Gaussian1}
Z_\mathrm{odd}(h) :=
\mathrm{tr} \exp[-\beta H_\mathrm{odd}(m, h)]
\le
Z_\mathrm{odd}(0).
\end{equation}
For any pairs of operators $A$ and $B$, we introduce the Duhamel two-point function by
\[
(A, B) := Z_\mathrm{odd}(0)^{-1}\int_0^1 ds \,
\mathrm{tr}\left( e^{-s \beta H_\mathrm{odd}(m, 0)} A e^{-(1-s)\beta H_\mathrm{odd}(m, 0) } B \right).
\]
Since $U_\mathrm{odd}^\dagger \rho(x) U_\mathrm{odd} = (-1)^{x_1 + \cdots + x_\nu} \rho(x)$, we see
\begin{align*}
U_\mathrm{odd}^\dagger H_{\mathrm{int}, \mu}^{(\Lambda)}(h) U_\mathrm{odd} 
&=
\frac{g}{2} \sum_{x \in \Lambda} [\rho(x) - \rho(x+e_\mu) +h^{(\mu)}(x)]^2\\
&=\frac{g}{2} \sum_{x \in \Lambda} \left[(\rho(x) - \rho(x+e_\mu))^2 
+2\rho(x) (h^{(\mu)}(x) - h^{(\mu)}(x-e_\mu))\right. \\
&\quad \quad \quad \quad \left. + h^{(\mu)}(x)^2\right],
\end{align*}
where $H_{\mathrm{int}, \mu}^{(\Lambda)}(h)$ is given by (\ref{Hintmuh}) 
for the interaction Hamiltonian $H_{\rm int}(h)$ of (\ref{Hinth}), and 
we have used $\sum_{x \in \Lambda} [\rho(x) - \rho(x+e_\mu)]h^{(\mu)}(x) 
= \sum_{x \in \Lambda}\rho(x) [h^{(\mu)}(x) - h^{(\mu)}(x-e_\mu)]$. 

We write $\partial_j h^{(\mu)} (x) := h^{(\mu)}(x) -h^{(\mu)}(x-e_j)$. 
We claim that the following inequality is valid:  
\begin{equation}
\label{Eq.infrared}
\left(\rho\left[\overline{\sum_\mu \partial_\mu h^{(\mu)}}\right], \rho\left[\sum_\mu \partial_\mu h^{(\mu)}\right] \right)
\le
\frac{1}{\beta g} \sum_{\mu=1}^\nu \sum_{x \in \Lambda} \vert h^{(\mu)}(x)\vert^2
\end{equation}
for any complex-valued functions $h^{(\mu)}$, 
where we have written  $\rho[f] := \sum_x \rho(x)f(x)$, and $\overline{z}$ denotes the complex conjugate of $z\in\mathbb{C}$. 
Using $d^2 Z_\mathrm{odd}(\epsilon h)/d\epsilon^2 \vert_{\epsilon = 0} \le 0$ by (\ref{Eq.Gaussian1})  and the identity
\[
{(A, A) Z_\mathrm{odd}(0) = \left.\frac{d^2}{d s^2} \mathrm{tr} \left[\exp(-\beta H_\mathrm{odd}(m, 0) + sA) \right] \right\vert_{s=0}},
\]
we have (\ref{Eq.infrared}) for real-valued functions $h^{(\mu)}$.
It also holds for complex-valued $h^{(\mu)}$, because $(A^\dagger, A) = (A_1, A_1) + (A_2, A_2)$  
for $A = A_1 + i A_2$ with $A_i^\dagger = A_i$, $i=1, 2$.
Thus the inequality (\ref{Eq.infrared}) holds for any complex-valued functions $h^{(\mu)}$. 

Taking $h^{(\mu)}(x) = \vert\Lambda\vert^{-1/2} (\exp(ip \cdot(x+ e_\mu)) - \exp(ip\cdot x))$ with $p= (p^{(1)}, \dots, p^{(\nu)})$, we have
\[
\partial_\mu h^{(\mu)}(x) = -2\vert\Lambda\vert^{-1/2}e^{ip\cdot x}(1- \cos p^{(\mu)})
\]
and
\[
\sum_{x \in \Lambda} \sum_{\mu=1}^\nu \vert h_m(x)\vert^2
=\sum_{\mu=1}^\nu(1- \cos p^{(\mu)}) =: {E_p}.
\]
For $p \in \Lambda^\ast$, the dual lattice of $\Lambda$, 
let $\tilde \rho_p := \vert\Lambda\vert^{-1/2} \sum_x \rho(x) \exp(ip \cdot x)$.
Then, from $H_\mathrm{odd}(m, h)=U_\mathrm{odd}^\dagger H^\mathrm{(\Lambda)}(m, h) U_\mathrm{odd}$ 
and the inequality (\ref{Eq.infrared}), we obtain the infrared bound
\begin{equation}
\label{Eq.IB}
(\tilde \rho_p, \tilde \rho_{-p})_{\beta, m} \le \frac{1}{2 \beta g E_{p+Q}},
\end{equation}
where $Q=(\pi, \dots, \pi)$ and
\[
(A, B)_{\beta, m} := \frac{1}{\mathrm{tr} \exp [-\beta H^{(\Lambda)}(m, 0)]}\int_0^1 ds \,
\mathrm{tr}\left( e^{-s \beta H^{(\Lambda)}(m, 0)} A e^{-(1-s)\beta H^{(\Lambda)}(m, 0) } B \right).
\]

Let $C_p := \langle [\tilde \rho_{p}, [H^{(\Lambda)}(0), \tilde \rho_{-p}]] \rangle_{\beta, 0}^{(\Lambda)} $ 
be the expectation value (\ref{expectationbetam}) of the double commutator with $m=0$. 
Then $C_p \ge 0$ by an eigenfunction expansion (see the next line of \cite[Eq.~(28)]{DLS}).
From~\cite[Thm.~3.2 \& Cor~3.2]{DLS} and the infrared bound (\ref{Eq.IB}), we obtain that
\begin{align*}
 \langle \tilde \rho_{p} \tilde \rho_{-p}  + \tilde \rho_{-p} \tilde \rho_{p} \rangle_{\beta, 0}^{(\Lambda)}
 &\le
 \sqrt{\frac{C_p}{2gE_{p+Q}}} \coth\left(\sqrt{\frac{C_p \beta^2 g E_{p+Q} }{2 } }\right) \\
 &\le
 \sqrt{ \frac{C_p}{2 g E_{p+Q}}} + \frac{1}{\beta g E_{p+Q}},
\end{align*}
where we have used $\coth x \le 1+ 1/x$.
Using this inequality, we have

\begin{equation}
\label{Eq.LRO}
\begin{split}
 \frac{1}{ \vert\Lambda\vert} \sum_{p \in \Lambda^\ast}    \langle \tilde \rho_{p} \tilde \rho_{-p}  + \tilde \rho_{-p} \tilde \rho_{p} \rangle_{\beta, 0}^{(\Lambda)}
&\le \vert\Lambda\vert^{-1}\sum_{p \neq Q}\left(  \frac{1}{  \beta g E_{p+Q}} 
+ \sqrt{\frac{C_p}{2g E_{p+Q}}}\right) 
+\frac{2\langle \tilde \rho_{Q} \tilde \rho_{Q}  \rangle_{\beta,  0}^{(\Lambda)}}{\vert\Lambda\vert}.
\end{split}
\end{equation}
The second term in the right-hand side of (\ref{Eq.LRO}) can be written in terms of the long-range order parameter 
which is defined by $m_\mathrm{LRO}^{(\Lambda)} :=  \vert\Lambda\vert^{-1}\sqrt{\langle [O^{(\Lambda)}]^2  \rangle_{\beta,  0}^{(\Lambda)}}$.
More precisely,
$(m_\mathrm{LRO}^{(\Lambda)})^2 = \vert\Lambda\vert^{-1} \langle  \tilde \rho_{Q} \tilde \rho_{Q} \rangle_{\beta,  0}^{(\Lambda)}$.

In order to estimate the first sum in the right-hand side of (\ref{Eq.LRO}), we have to evaluate the double commutator in $C_p$. 
Since $[\rho(x), \rho(y)] = 0$ we have
\[
[\tilde \rho_{p}, [H^{(\Lambda)}(0), \tilde \rho_{-p}]] 
 =
[\tilde \rho_{p}, [H_K^{(\Lambda)}, \tilde \rho_{-p}]],
\]
where $H_K^{(\Lambda)}$ is the free part of the present Hamiltonian $H^{(\Lambda)}(m)$ with the mass parameter $m=0$. 
Using the commutation relations, we see
\begin{align*}
\left\| [\tilde \rho_{p}, [H_K^{(\Lambda)}, \tilde \rho_{-p}]]    \right\|
&\le
4\vert\kappa\vert v \|\psi^\dagger(x) \psi(y) - \psi^\dagger(y)\psi(x)  \|\\
&\le
8\vert\kappa\vert v.
\end{align*}
Hence $C_p \le 8\vert\kappa\vert v$ follows.
Then we have
\begin{align*}
 \frac{1}{\vert\Lambda\vert}\sum_{p \neq Q} \sqrt{\frac{C_p}{2  g E_{p+Q}}}  
 &\le
\frac{2}{\vert\Lambda\vert} \sum_{p \neq Q} \sqrt{\frac{\vert\kappa\vert \nu}{g E_{p+Q}}}.
\end{align*}
From the Plancherel theorem and $\rho(x)^2 = 1/4$, we obtain
\begin{equation}
\label{sumrule}
\sum_{p \in \Lambda^\ast}  \langle \tilde \rho_{p} \tilde \rho_{-p}  
+ \tilde \rho_{-p} \tilde \rho_{p} \rangle_{\beta, 0}^{(\Lambda)}= \frac{\vert\Lambda\vert}{2}.
\end{equation}
Putting all together, we have
\[
\frac{1}{4}
\le \frac{I_\nu}{2\beta g} + \sqrt{\frac{\vert\kappa\vert \nu}{g} }J_\nu 
+\underbrace{\lim_{\Lambda \nearrow \mathbb Z^{\nu}}(m^{(\Lambda)}_\mathrm{LRO})^2}_{=: m_\mathrm{LRO}^2},
\]
with  $I_\nu$ and $J_\nu$ given by
\[
I_\nu := \frac{1}{(2\pi)^\nu}\int_{[-\pi, \pi]^\nu} \frac{dp}{E_p}, \quad
J_\nu := \frac{1}{(2\pi)^\nu}\int_{[-\pi, \pi]^\nu} \frac{dp}{\sqrt{E_p}}.
\]
We note that $I_\nu$ and $J_\nu$ are finite for $\nu \ge 3$.
Hence the long-range order $m_\mathrm{LRO} >0$ exists for sufficiently large $\beta$ and small $\vert\kappa\vert/g$. \qed

\section{Long-range order at zero temperature}
\label{LROzero}

In this section, we prove the existence of long-range order in the ground states \cite{Neves1986LongRO, KLS1, KLS2} 
for dimensions $\nu \ge 2$. 

Let $\{\Psi^{(\Lambda)}_i \}_{i=1}^\mu$ be the ground states for $H^{(\Lambda)}(m)$ with $\mu$-multiplicity.
Then it is known that, as $\beta \to \infty$, the thermal expectation value of any observable $A$ converges to 
the expectation value in the ground states:
\[
\lim_{\beta \to \infty} \langle  A  \rangle_{\beta,  0}^{(\Lambda)}
=
\frac{1}{\mu}\sum_{i=1}^\mu \langle \Psi^{(\Lambda)}_i , A  \Psi^{(\Lambda)}_i  \rangle 
=:
\omega_0^{(\Lambda)}(A ).
\]
Taking $\beta \to \infty$ in (\ref{Eq.LRO}), we have
\[
\frac{1}{4}
\le
\frac{1}{\vert\Lambda\vert} \sum_{p \neq Q} \sqrt{\frac{\vert\kappa\vert \nu}{g E_{p+Q}}}
+ \frac{\omega_0^{(\Lambda)}(\tilde \rho_{Q} \tilde \rho_{Q} )}{\vert\Lambda\vert},
\]
where we have used $C_p \le 8 \vert\kappa\vert \nu$ and the sum rule (\ref{sumrule}).
Hence in the infinite volume limit, we obtain
\[
m^2_\mathrm{GSLRO}
:=
\lim_{\Lambda \nearrow \mathbb Z^\nu} \frac{\omega_0^{(\Lambda)}(\tilde \rho_{Q} \tilde \rho_{Q} )}{\vert\Lambda\vert}
\ge
\frac{1}{4} - \sqrt{\frac{\vert\kappa\vert\nu}{g}} J_\nu.
\]
Since $J_\nu < \infty$ for $\nu \ge 2$, the long-range order $m_\mathrm{GSLRO} >0$ exists for sufficiently small $\vert\kappa\vert/g$. \qed

\appendix 

\section{Chirality and mass for the staggered fermions}
\label{staggeredfermion}

Following \cite{BRSKJSS}, we review the chirality of the free fermion part of the Hamiltonian $H^{(\Lambda)}(m)$ 
of (\ref{HamLambdam}) in the case of the three dimension $\nu=3$. 
{For a more careful mathematical discussion on the staggered fermions, we refer the reader to~\cite{Nakamura}.} 

Let us consider the Hamiltonian of \cite{BRSKJSS} with the mass term which is given by 
\begin{eqnarray}
\label{Hamfree}
H_{\rm free}^{(\Lambda)}&:=&i\kappa \sum_{x\in\Lambda\subset\mathbb{Z}^3}
\Bigl\{[\psi^\dagger(x)\psi(x+e_1)-\psi^\dagger(x+e_1)\psi(x)]\nonumber\\
&+&i[\psi^\dagger(x)\psi(x+e_2)+\psi^\dagger(x+e_2)\psi(x)](-1)^{x^{(1)}+x^{(2)}}\nonumber\\
&+&[\psi^\dagger(x)\psi(x+e_3)-\psi^\dagger(x+e_3)\psi(x)](-1)^{x^{(1)}+x^{(2)}}\Bigr\}\nonumber\\
&+&m\sum_{x\in\Lambda\subset\mathbb{Z}^3}(-1)^{x^{(1)}+x^{(2)}+x^{(3)}}\psi^\dagger(x)\psi(x).
\end{eqnarray}
Here, we impose the anti-periodic boundary conditions, 
\begin{equation*}
\psi(L+1,x^{(2)},x^{(3)})=-\psi(-L+1,x^{(2)},x^{(3)}),
\end{equation*}
\begin{equation*}
\psi(x^{(1)},L+1,x^{(3)})=-\psi(x^{(1)},-L+1,x^{(3)}),
\end{equation*}
and 
\begin{equation*}
\psi(x^{(1)},x^{(2)},L+1)=-\psi(x^{(1)},x^{(2)},-L+1)
\end{equation*}
for the fermion operator $\psi(x)$ at the right boundaries. Similarly, 
\begin{equation*}
\psi(-L,x^{(2)},x^{(3)})=-\psi(L,x^{(2)},x^{(3)}),
\end{equation*}
\begin{equation*}
\psi(x^{(1)},-L,x^{(3)})=-\psi(x^{(1)},L,x^{(3)}),
\end{equation*}
and 
\begin{equation*}
\psi(x^{(1)},x^{(2)},-L)=-\psi(x^{(1)},x^{(2)},L)
\end{equation*}
for the left boundaries. These are equivalent to the anti-periodic boundary conditions for the Hamiltonian $H^{(\Lambda)}(m)$. 
These conditions are realized by the Fourier transform of $\psi(x)$, 
\begin{equation}
\psi(x)=\frac{1}{\sqrt{\vert\Lambda\vert}}\sum_k e^{ikx}\hat{\psi}(k),
\end{equation}
with the momenta $k=(k^{(1)},k^{(2)},k^{(3)})$ which satisfy 
\begin{equation}
\label{APBC}
\exp[ik^{(i)}\cdot 2L]=-1 \quad \mbox{for } i=1,2,3.
\end{equation}

The second term in the summand of the first sum in the right-hand side of (\ref{Hamfree}) 
is different from the corresponding term of the Hamiltonian $H^{(\Lambda)}(m)$ of (\ref{HamLambdam}). 
In order to show that the two Hamiltonians are unitary equivalent to each other, we introduce a unitary transformation, 
\begin{equation}
U_{\rm free}:=\prod_{x\;:\;x^{(2)}={\rm odd}}\exp[-(i\pi/2)\psi^\dagger(x)\psi(x)]. 
\end{equation}
Then, one has 
\begin{equation}
U_{\rm free}^\dagger \psi(x) U_{\rm free}=
\begin{cases}
-i\psi(x) & \mbox{for } x^{(2)}={\rm odd};\\
\psi(x) & \mbox{for } x^{(2)}={\rm even}.
\end{cases}
\end{equation}
Under this transformation, only the second term in the summand of the first sum in the right-hand side of (\ref{Hamfree}) 
changes as follows: 
\begin{eqnarray}
& &U_{\rm free}^\dagger [\psi^\dagger(x)\psi(x+e_2)+\psi^\dagger(x+e_2)\psi(x)]U_{\rm free}\nonumber\\
&=&
\begin{cases}
i\psi^\dagger(x)\psi(x+e_2)-i\psi^\dagger(x+e_2)\psi(x) & \mbox{for } x^{(2)}={\rm odd};\\
-i\psi^\dagger(x)\psi(x+e_2)+i\psi^\dagger(x+e_2)\psi(x) & \mbox{for } x^{(2)}={\rm even}
\end{cases}\nonumber\\
&=&-i(-1)^{x^{(2)}}[\psi^\dagger(x)\psi(x+e_2)-\psi^\dagger(x+e_2)\psi(x)]. 
\end{eqnarray}
Therefore, the transformed Hamiltonian $U_{\rm free}^\dagger H_{\rm free}^{(\Lambda)}U_{\rm free}$ has the desired form, 
which is equal to the free part of the Hamiltonian $H^{(\Lambda)}(m)$ of (\ref{HamLambdam}). 

Since the Hamiltonian $H_{\rm free}^{(\Lambda)}$ of (\ref{Hamfree}) has the quadratic form in the fermion operators $\psi(x)$, 
it can be diagonalized in terms of fermion operators $\mathcal{A}_i$ as follows: 
\begin{equation}
H_{\rm free}^{(\Lambda)}=\sum_i E_i \mathcal{A}_i^\dagger \mathcal{A}_i,
\end{equation}
where $E_i$ are the energy eigenvalues. One has 
\begin{equation}
[H_{\rm free}^{(\Lambda)},\mathcal{A}_i^\dagger]=E_i\mathcal{A}_i^\dagger
\end{equation}
by using the anticommutation relations for $\mathcal{A}_i$. In order to find the solution $\mathcal{A}^\dagger$ 
for the equation $[H_{\rm free}^{(\Lambda)},\mathcal{A}^\dagger]=E\mathcal{A}^\dagger$, we set 
\begin{equation}
\mathcal{A}^\dagger =\sum_{x\in\Lambda} v(x)\psi^\dagger(x),
\end{equation}
where $v(x)$ is a function of the site $x\in\Lambda$. We choose the function $v(x)$ so that $v(x)$ satisfies 
the same anti-periodic boundary conditions as those of $\psi(x)$. 
Substituting this into $[H_{\rm free}^{(\Lambda)},\mathcal{A}^\dagger]=E\mathcal{A}^\dagger$, one has 
\begin{eqnarray}
& &\sum_{x\in\Lambda} [H_{\rm free}^{(\Lambda)},v(x)\psi^\dagger(x)]\nonumber\\
&=&i\kappa\sum_{x\in\Lambda} v(x)\Bigl\{[\psi^\dagger(x-e_1)-\psi^\dagger(x+e_1)]
+i(-1)^{x^{(1)}+x^{(2)}}[-\psi^\dagger(x-e_2)+\psi^\dagger(x+e_2)]\nonumber\\
&+&(-1)^{x^{(1)}+x^{(2)}}[\psi^\dagger(x-e_3)-\psi^\dagger(x+e_3)]\Bigr\}
+m\sum_{x\in\Lambda}v(x)(-1)^{x^{(1)}+x^{(2)}+x^{(3)}}\psi^\dagger(x)\nonumber\\
&=&E\sum_{x\in\Lambda}v(x)\psi^\dagger(x). 
\end{eqnarray}
{From} the coefficients of $\psi^\dagger(x)$, one obtains the eigenvalue equation, 
\begin{eqnarray}
\label{eigenvalueEq}
& &i\kappa\bigl\{[v(x+e_1)-v(x-e_1)]+i(-1)^{x^{(1)}+x^{(2)}}[v(x+e_2)-v(x-e_2)]\nonumber\\
&+&(-1)^{x^{(1)}+x^{(2)}}[v(x+e_3)-v(x-e_3)       \bigr\}+m(-1)^{x^{(1)}+x^{(2)}+x^{(3)}}v(x)
=Ev(x),\nonumber\\
\end{eqnarray}
for the single particle wavefunction $v(x)$. 

For the function $v(x)$, we introduce the Fourier transform, 
\begin{equation}
v(x)=\frac{1}{\sqrt{\vert\Lambda\vert}}\sum_k e^{ikx}\hat{v}(k),
\end{equation}
with the momentum $k=(k^{(1)},k^{(2)},k^{(3)})\in (-\pi,\pi]^3$. Substituting this into the above equation (\ref{eigenvalueEq}), 
one has 
\begin{equation}
\label{eigenvalueEqvk}
-2\kappa\bigl[\sin k^{(1)}\hat{v}(k)+i\sin k^{(2)}\hat{v}(k+\pi_{12})+\sin k^{(3)}\hat{v}(k+\pi_{12})\bigr]
+m\hat{v}(k+\pi_{123})=E\hat{v}(k), 
\end{equation} 
where we have written 
$$
\pi_{12}:=(\pi,\pi,0) \quad \mbox{and}\quad \pi_{123}:=(\pi,\pi,\pi). 
$$
We also write 
$$
\pi_1:=(\pi,0,0), \ \pi_2:=(0,\pi,0), \ \pi_3:=(0,0,\pi)   
$$
and 
$$
\pi_{13}:=(\pi,0,\pi), \ \pi_{23}:=(0,\pi,\pi).
$$
When $m=0$, these momenta as well as $k=0$ yield the zero energy $E=0$. 
Following \cite{BRSKJSS}, we define
\begin{equation}
\label{vuk}
\begin{pmatrix}
\hat{v}_{\rm u}^{(1)}(k)\\ \hat{v}_{\rm u}^{(2)}(k)\\ \hat{v}_{\rm u}^{(3)}(k)\\ \hat{v}_{\rm u}^{(4)}(k)
\end{pmatrix}
:=U_{\rm u}
\begin{pmatrix} \hat{v}(k)\\ \hat{v}(k+\pi_3)\\ \hat{v}(k+\pi_{12})\\ \hat{v}(k+\pi_{123})\end{pmatrix}
\end{equation}
for the {u-}quark, where we have written 
\begin{equation}
\label{Uu}
U_{\rm u}:=\frac{1}{2}\begin{pmatrix}
1 & 1 & 1 & 1 \\ 1 & -1 & -1 & 1 \\ 1 & -1 & 1 & -1 \\ 1 & 1 & -1 & -1 
\end{pmatrix}.
\end{equation}
{Clearly, the decomposition in the above equation (\ref{vuk}) is {justified} only for a small $k$.
However, one can expect that the contributions for large $k$ do not affect the continuum limit. 
In fact, in the sence of the norm for the resolvent of the Dirac operator, the statement was proved by \cite{Nakamura}.}

Similarly, for the {d-}quark, 
\begin{equation}
\begin{pmatrix}
\hat{v}_{\rm d}^{(1)}(k)\\ \hat{v}_{\rm d}^{(2)}(k)\\ \hat{v}_{\rm d}^{(3)}(k)\\ \hat{v}_{\rm d}^{(4)}(k)
\end{pmatrix}
:=U_{\rm d}
\begin{pmatrix}
\hat{v}(k+\pi_1)\\ \hat{v}(k+\pi_2) \\ \hat{v}(k+\pi_{23})\\ \hat{v}(k+\pi_{13})
\end{pmatrix},
\end{equation}
where 
\begin{equation}
U_{\rm d}:=\frac{1}{2}\begin{pmatrix}
-1 & 1 & -1 & 1 \\ -1 & -1 & -1 & -1 \\ 1 & -1 & -1 & 1 \\ 1 & 1 & -1 & -1 
\end{pmatrix}.
\end{equation}
We write
$$
\hat{v}_{\rm f}(k):=\begin{pmatrix}
\hat{v}_{\rm f}^{(1)}(k) \\ \hat{v}_{\rm f}^{(2)}(k) \\ \hat{v}_{\rm f}^{(3)}(k) \\ \hat{v}_{\rm f}^{(4)}(k)
\end{pmatrix}
$$
for ${\rm f}={\rm u,d}$. We also write 
\begin{equation*}
\sigma^{(1)}:=\begin{pmatrix}
0 & 1 \\
1 & 0 
\end{pmatrix},\quad 
\sigma^{(2)}:=\begin{pmatrix}
0 & -i \\
i & 0 
\end{pmatrix},\quad 
\sigma^{(3)}:=\begin{pmatrix}
1 & 0 \\
0 & -1 
\end{pmatrix}, 
\end{equation*}
$$
\sin k:=(\sin k^{(1)},\sin k^{(2)},\sin k^{(3)}),
$$ 
and 
$$
\sin k\cdot \sigma :=\sin k^{(1)}\sigma^{(1)}+\sin k^{(2)}\sigma^{(2)}+\sin k^{(3)}\sigma^{(3)}.
$$
By using these expressions and the eigenvalue equation (\ref{eigenvalueEqvk}), one has 
\begin{equation}
\left[-2\kappa
\begin{pmatrix}
0 & \sin k\cdot \sigma \\
\sin k\cdot \sigma & 0
\end{pmatrix}
+m\gamma_0 \right]\hat{v}_{\rm f}(k)=E\hat{v}_{\rm f}(k) 
\end{equation}
for low momenta $k$, where we have written 
$$
\gamma_0:=\begin{pmatrix}
1 & 0 & 0 & 0\\
0 & 1 & 0 & 0\\
0 & 0 & -1 & 0 \\
0 & 0 & 0 & -1 
\end{pmatrix}.
$$
This is nothing but the desired Dirac equation for small momenta $k$. 

In order to discuss the chirality in the case with $m=0$, we define  
\begin{equation}
\label{gamma5}
\gamma_5:=\begin{pmatrix}
0 & 0 & 1 & 0 \\
0 & 0 & 0 & 1 \\
1 & 0 & 0 & 0 \\
0 & 1 & 0 & 0 
\end{pmatrix}. 
\end{equation}
Then, one has 
\begin{equation}
\Bigl[\gamma_5,\begin{pmatrix} 0 & \sigma^{(i)} \\ \sigma^{(i)} & 0 \end{pmatrix}\Bigr]=0. 
\end{equation}
for all the indices $i=1,2,3$. This implies that $\gamma_5$ is the conserved quantity, and its eigenvalues 
take $\pm 1$, i.e., the chirality. Since $\gamma_5$ does not commute with $\gamma_0$, it is not conserved 
in the case with $m\ne 0$. Actually, one has the following important relation: 
\begin{equation}
\label{gamma50relation}
\gamma_5 \gamma_0 \gamma_5=-\gamma_0.
\end{equation}

As mentioned above, we have ignored the contributions from large momenta $k$.  
Then, in the case of $m=0$, the chirality is conserved. This procedure is justified in the sense of the norm 
for the resolvent of the Dirac operator \cite{Nakamura}. Therefore, we can expect that the free fermion Hamiltonian is 
chiral invariant in the continuum limit with $m=0$, although we will not give the precise proof in the present paper.

\section{Chiral invariance of the interaction Hamiltonian {under a low-energy approximation}}
\label{Sec:chiralint}

In this appendix, we show that the interaction Hamiltonian of $H^{(\Lambda)}(m)$ of (\ref{HamLambdam}) is invariant under 
a certain chiral transformation {under a low-energy approximation}. 
The chiral transformation is discrete in the sense of \cite{BRSKJSS}. 
As mentioned at the end of the preceding section, we have ignored the contributions from large momenta $k$ 
in order to show the chiral invariance of the free fermion Hamiltonian with the vanishing mass parameter, $m=0$. 
For the purpose of the present section, we will also use this approximation. 
We hope that the following argument for the continuum limit will be justified in future studies. 
In passing, we stress that our main results for the phase transition are still mathematically rigorous, 
because we have proved the particle-hole symmetry breaking in the main part of the present paper.

To begin with, we introduce the following Fourier transform: 
\begin{equation}
\label{Fourierpsi}
\psi(x)=\frac{1}{\sqrt{\vert\Lambda\vert}}e^{i\pi x^{(2)}}\sum_k e^{ikx}\hat{\psi}(k). 
\end{equation}
We consider first the free part of the Hamiltonian. 
Substituting the Fourier transform (\ref{Fourierpsi}) into the expression (\ref{Hamfree}) 
of the Hamiltonian $H_{\rm free}^{(\Lambda)}$, one has 
\begin{eqnarray}
H_{\rm free}^{(\Lambda)}
&=&\sum_k \Bigl\{-2\kappa[\sin k^{(1)}\hat{\psi}^\dagger(k)\hat{\psi}(k)
+i\sin k^{(2)}\hat{\psi}^\dagger(k)\hat{\psi}(k+\pi_{12})\nonumber\\
&+&\sin k^{(3)}\hat{\psi}^\dagger(k)\hat{\psi}(k+\pi_{12})]+m\hat{\psi}^\dagger(k)\hat{\psi}(k+\pi_{123})\Bigr\}\nonumber\\
&=&H_{\rm free,u}^{(\Lambda)}+H_{\rm free,d}^{(\Lambda)}, 
\end{eqnarray}
where we have written 
\begin{eqnarray}
H_{\rm free,u}^{(\Lambda)}&=&\sum_k \Bigl\{-\kappa[\sin k^{(1)}\hat{\psi}^\dagger(k)\hat{\psi}(k)
+i\sin k^{(2)}\hat{\psi}^\dagger(k)\hat{\psi}(k+\pi_{12})\nonumber\\
&+&\sin k^{(3)}\hat{\psi}^\dagger(k)\hat{\psi}(k+\pi_{12})]
+\frac{m}{2}\hat{\psi}^\dagger(k)\hat{\psi}(k+\pi_{123})\Bigl\}
\end{eqnarray}
and 
\begin{eqnarray}
H_{\rm free,d}^{(\Lambda)}&=&\sum_k \Bigl\{-\kappa[-\sin k^{(1)}\hat{\psi}^\dagger(k+\pi_1)\hat{\psi}(k+\pi_1)
+i\sin k^{(2)}\hat{\psi}^\dagger(k+\pi_1)\hat{\psi}(k+\pi_{2})\nonumber\\
&+&\sin k^{(3)}\hat{\psi}^\dagger(k+\pi_1)\hat{\psi}(k+\pi_{2})]
+\frac{m}{2}\hat{\psi}^\dagger(k+\pi_1)\hat{\psi}(k+\pi_{23})\Bigl\}. 
\end{eqnarray}
We write 
\begin{equation}
\label{tildePsi}
\tilde{\Psi}_{\rm u}(k):=
\begin{pmatrix}
\hat{\psi}(k) \\ \hat{\psi}(k+\pi_3)\\ \hat{\psi}(k+\pi_{12})\\ \hat{\psi}(k+\pi_{123})
\end{pmatrix}.
\end{equation}
Then, one has 
\begin{eqnarray}
H_{\rm free,u}^{(\Lambda)}
&=&\frac{1}{4}\sum_k \tilde{\Psi}_{\rm u}^\dagger(k)\Bigl\{
-\kappa\Bigl[\sin k^{(1)}
\begin{pmatrix} 
1 & 0 \\ 0 & -1
\end{pmatrix}
+\sin k^{(2)}
\begin{pmatrix}
0 & i \\ -i & 0
\end{pmatrix}
+\sin k^{(3)}
\begin{pmatrix}
0 & \sigma^{(3)} \\ \sigma^{(3)} & 0 
\end{pmatrix}
\Bigr]\nonumber\\
&+&\frac{m}{2}\begin{pmatrix}
0 & \sigma^{(1)}\\ \sigma^{(1)} & 0
\end{pmatrix}
\Bigr\}\tilde{\Psi}_{\rm u}(k).
\end{eqnarray}
Further, we write 
\begin{equation}
\label{hatPsiu}
\hat{\Psi}_{\rm u}(k):=\begin{pmatrix} \hat{\Psi}_{\rm u}^{(1)}(k) \\ \hat{\Psi}_{\rm u}^{(2)}(k) \\
\hat{\Psi}_{\rm u}^{(3)}(k) \\ \hat{\Psi}_{\rm u}^{(4)}(k)\end{pmatrix}
:=U_{\rm u}\tilde{\Psi}_{\rm u}(k) 
\end{equation}
in terms of $U_{\rm u}$ of (\ref{Uu}). Then, the above Hamiltonian has the desired form, 
\begin{equation}
H_{\rm free,u}^{(\Lambda)}=\frac{1}{4}\sum_k \hat{\Psi}_{\rm u}^\dagger(k)\Bigl[
-\kappa \begin{pmatrix}
0 & \sin k\cdot \sigma \\ \sin k\cdot \sigma & 0 
\end{pmatrix}
+\frac{m}{2}\gamma_0\Bigr]\hat{\Psi}_{\rm u}(k). 
\end{equation}
However, the momenta $k$ should be restricted to small values so that the components of the operators $\hat{\Psi}_{\rm u}(k)$ are 
independent {of} each other. 

Similarly, 
\begin{equation}
\label{tildePsid}
\tilde{\Psi}_{\rm d}(k):=\begin{pmatrix}
\hat{\psi}(k+\pi_1) \\ \hat{\psi}(k+\pi_2) \\ \hat{\psi}(k+\pi_{23}) \\ \hat{\psi}(k+\pi_{13}) 
\end{pmatrix}.
\end{equation}
Then, 
\begin{eqnarray}
H_{\rm free,d}^{(\Lambda)}
&=&\frac{1}{4}\sum_k \tilde{\Psi}_{\rm d}^\dagger(k)\Bigl\{
-\kappa\Bigl[\sin k^{(1)}
\begin{pmatrix} 
-\sigma^{(3)} & 0 \\ 0 & \sigma^{(3)}
\end{pmatrix}
+\sin k^{(2)}
\begin{pmatrix}
-\sigma^{(2)} & 0 \\ 0 & \sigma^{(2)}
\end{pmatrix}\nonumber\\
&+&\sin k^{(3)}
\begin{pmatrix}
\sigma^{(1)} & 0 \\ 0 & -\sigma^{(1)} 
\end{pmatrix}
\Bigr]
+\frac{m}{2}\begin{pmatrix}
0 & 1 \\ 1 & 0
\end{pmatrix}
\Bigr\}\tilde{\Psi}_{\rm d}(k).
\end{eqnarray}
Further, we write 
\begin{equation}
\label{hatPsid}
\hat{\Psi}_{\rm d}(k):=\begin{pmatrix} \hat{\Psi}_{\rm d}^{(1)}(k) \\ \hat{\Psi}_{\rm d}^{(2)}(k) \\
\hat{\Psi}_{\rm d}^{(3)}(k) \\ \hat{\Psi}_{\rm d}^{(4)}(k)\end{pmatrix}
:=U_{\rm d}\tilde{\Psi}_{\rm d}(k).  
\end{equation}
These yield the same form of the Hamiltonian,   
\begin{equation}
H_{\rm free,d}^{(\Lambda)}=\frac{1}{4}\sum_k \hat{\Psi}_{\rm d}^\dagger(k)\Bigl[
-\kappa \begin{pmatrix}
0 & \sin k\cdot \sigma \\ \sin k\cdot \sigma & 0 
\end{pmatrix}
+\frac{m}{2}\gamma_0\Bigr]\hat{\Psi}_{\rm d}(k). 
\end{equation}

Now let us consider the interaction term of the Hamiltonian $H^{(\Lambda)}(m)$ of (\ref{HamLambdam}). 
{From} the expression (\ref{Fourierpsi}), one has  
\begin{eqnarray}
\label{HamintFourier}
& &\sum_{x\in\Lambda} \psi^\dagger(x)\psi(x)\psi^\dagger(x+e_\mu)\psi(x+e_\mu)\nonumber\\
&=&\frac{1}{\vert\Lambda\vert^2}\sum_{x\in\Lambda}\sum_{k_1,k_2,k_3,k_4}e^{-ik_1x}e^{ik_2x}e^{-ik_3(x+e_\mu)}e^{ik_4(x+e_\mu)}
\hat{\psi}^\dagger(k_1)\hat{\psi}(k_2)\hat{\psi}^\dagger(k_3)\hat{\psi}(k_4)\nonumber\\
&=&\frac{1}{\vert\Lambda\vert}\sum_{k_1,k_2,k_3,k_4}\delta_{k_1+k_3,k_2+k_4}e^{i(k_4-k_3)e_\mu}
\hat{\psi}^\dagger(k_1)\hat{\psi}(k_2)\hat{\psi}^\dagger(k_3)\hat{\psi}(k_4). 
\end{eqnarray}
In order to treat the region of the low energies, we consider the corresponding momenta,  
$$
k_i=p_i+K_i \ \ \mbox{for } i=1,2,3,4,
$$
where $p_i$ are small momenta, and $K_i$ are given by 
\begin{equation}
\label{Ki}
K_i=(K_i^{(1)},K_i^{(2)},K_i^{(3)})
\end{equation} 
with $K_i^{(j)}\in\{0,\pi\}$ for $i=1,2,3,4$ and $j=1,2,3$. 
The momentum conservation of the Kronecker delta implies 
\begin{equation}
\label{conservemoment}
p_1+p_3=p_2+p_4 \quad \mbox{and}\quad K_1+K_3=K_2+K_4. 
\end{equation}

Let $K=(K^{(1)},K^{(2)},K^{(3)})$ for $K^{(i)}\in\{0,\pi\}$, $i=1,2,3$. 
We define a signature ${\rm sgn}(K)$ for $K$ as follows: 
\begin{equation}
{\rm sgn}(K):=(-1)^{\iota(K^{(1)})+\iota(K^{(2)})+\iota(K^{(3)})},
\end{equation}
where 
\begin{equation}
\iota(K^{(i)}):=\begin{cases}
1, & \mbox{if} \ K^{(i)}=\pi;\\
0, & \mbox{if} \ K^{(i)}=0
\end{cases}
\end{equation}
for $i=1,2,3$. 
Then, the second relation of the momentum conservation law (\ref{conservemoment}) implies 
\begin{equation}
\label{K4relation}
{\rm sgn}(K_1){\rm sgn}(K_2){\rm sgn}(K_3){\rm sgn}(K_4)=+1.
\end{equation}

Having these observations in mind, let us consider the chiral transformation $\gamma_5$. 
{From} (\ref{Uu}), (\ref{gamma5}), (\ref{tildePsi}) and (\ref{hatPsiu}), we have 
\begin{eqnarray}
U_{\rm u}^\dagger \gamma_5 U_{\rm u}\tilde{\Psi}_{\rm u}&=&
\begin{pmatrix}
1 & 0 & 0 & 0 \\ 0 & -1 & 0 & 0 \\ 0 & 0 & 1 & 0 \\ 0 & 0 & 0 & -1 
\end{pmatrix}
\begin{pmatrix}
\hat{\psi}(k) \\ \hat{\psi}(k+\pi_3)\\ \hat{\psi}(k+\pi_{12})\\ \hat{\psi}(k+\pi_{123})
\end{pmatrix}\nonumber\\
&=&
\begin{pmatrix}
\hat{\psi}(k) \\ -\hat{\psi}(k+\pi_3)\\ \hat{\psi}(k+\pi_{12})\\ -\hat{\psi}(k+\pi_{123})
\end{pmatrix}
=\begin{pmatrix}
(\gamma_5\hat{\psi})(k) \\ (\gamma_5\hat{\psi})(k+\pi_3)\\ (\gamma_5\hat{\psi})(k+\pi_{12})\\ (\gamma_5\hat{\psi})(k+\pi_{123})
\end{pmatrix},  
\end{eqnarray}
where we have written 
\begin{equation}
\label{gamma5psi}
(\gamma_5\hat{\psi})(k+K)={\rm sgn}(K)\hat{\psi}(k+K)
\end{equation}
Similarly, 
\begin{equation}
U_{\rm d}^\dagger \gamma_5 U_{\rm d}\tilde{\Psi}_{\rm d}=
\begin{pmatrix}
-\hat{\psi}(k+\pi_1) \\ -\hat{\psi}(k+\pi_2) \\ \hat{\psi}(k+\pi_{23}) \\ \hat{\psi}(k+\pi_{13})
\end{pmatrix}
=
\begin{pmatrix}
(\gamma_5\hat{\psi})(k+\pi_1) \\ (\gamma_5\hat{\psi})(k+\pi_2) \\ (\gamma_5\hat{\psi})(k+\pi_{23}) \\ 
(\gamma_5\hat{\psi})(k+\pi_{13})
\end{pmatrix}.
\end{equation}

The operator in the right-hand side of the interaction Hamiltonian (\ref{HamintFourier}) is transformed 
under the transformation $\gamma_5$ as follows: 
\begin{eqnarray}
& &\hat{\psi}^\dagger(p_1+K_1)\hat{\psi}(p_2+K_2)\hat{\psi}^\dagger(p_3+K_3)\hat{\psi}(p_4+K_4)\nonumber\\
&\rightarrow&  (\gamma_5\hat{\psi})^\dagger(p_1+K_1)(\gamma_5\hat{\psi})(p_2+K_2)(\gamma_5\hat{\psi})^\dagger(p_3+K_3)
(\gamma_5\hat{\psi})(p_4+K_4)\nonumber\\
&=&{\rm sgn}(K_1){\rm sgn}(K_2){\rm sgn}(K_3){\rm sgn}(K_4)
\hat{\psi}^\dagger(p_1+K_1)\hat{\psi}(p_2+K_2)\hat{\psi}^\dagger(p_3+K_3)\hat{\psi}(p_4+K_4)\nonumber\\
&=&\hat{\psi}^\dagger(p_1+K_1)\hat{\psi}(p_2+K_2)\hat{\psi}^\dagger(p_3+K_3)\hat{\psi}(p_4+K_4),
\end{eqnarray}
where we have used (\ref{K4relation}) and (\ref{gamma5psi}). Thus, the interaction Hamiltonian is invariant 
under the discrete chiral transformation $\gamma_5$ when the energies are restricted to the low values. 
This implies that, under the same restriction, the total Hamiltonian without the mass term is invariant 
under the discrete chiral transformation $\gamma_5$.   

\section{Chiral symmetry breaking in an infinite-volume and continuum limit}
\label{Dlimit}

In this appendix, we discuss the infinite-volume and continuum limit. 
For the present Hamiltonian, we set $m=0$, i.e., we consider the case without the mass term in the Hamiltonian. 
We also assume that the contributions from the high energies can be neglected in the continuum limit. 
{Namely, we will continue to use the low-energy approximation. 
Therefore, the discussion here is fairly formal, but we believe it will be helpful to the reader for comprehending 
the physics in the infinite-volume and continuum limit. In particular, the form (\ref{StandardInt}) of an interaction below 
is standard, although one might think that interactions generally exhibit a very complicated form in the representation 
by using the staggered fermions.}
\subsection{Chiral symmetry breaking}

We write $\langle\cdots\rangle$ for the expectation value in the infinite-volume and continuum limit. 
Let us deal with the case that the expectation value $\langle\cdots\rangle$ is invariant under the discrete chiral 
transformation $\gamma_5$. Consider the two-point correlation function,
\begin{equation}
\langle \hat{\psi}^\dagger(p_1+K_1)\hat{\psi}(p_2+K_2)\rangle. 
\end{equation}
{From} the argument of the preceding Appendix~\ref{Sec:chiralint}, we obtain the following result: 
When ${\rm sgn}(K_1){\rm sgn}(K_2)=-1$, this expectation value is vanishing under the above assumption. 
Clearly, generic $n$-point correlation functions have the same property as well.   

Similarly, from the relation (\ref{gamma50relation}), one can show that the expectation value of the mass term of the Hamiltonian 
is also vanishing under the same assumption on the expectation value. 
Conversely, if the expectation value of the mass term shows a non-vanishing value, 
then the chiral symmetry must be broken.  

\subsection{The interaction Hamiltonian in an infinite-volume and continuum limit}

If one expresses the interactions in the right-hand of (\ref{HamintFourier}) in terms of 
the operators of the {u- and d-} quarks, then the expressions become very complicated. 
The reason is that the interaction processes allow the momentum transfer given by (\ref{Ki}).  
In order to clarify this point, we consider more general interactions which have the following form: 
\begin{equation}
\label{generalIntG}
\sum_{x,y\in\Lambda}\psi^\dagger(x)\psi(x)G^{(\Lambda)}(x-y)(-1)^{x^{(1)}+x^{(2)}+x^{(3)}+y^{(1)}+y^{(2)}+y^{(3)}}
\psi^\dagger(y)\psi(y),
\end{equation}
where the function $G^{(\Lambda)}(x)$ of the site $x\in\Lambda$ is given by 
\begin{equation}
\label{G}
G^{(\Lambda)}(x):=\frac{1}{\vert\Lambda\vert}\sum_k \hat{G}(k)e^{ikx}
\end{equation}
with a real-valued function $\hat{G}(k)$ of the momentum $k$. 
Here, we require the condition ${\rm supp}\; \hat{G}(k)\subset [-\varepsilon_0,\varepsilon_0]^3$ 
with a small positive $\varepsilon_0$. This condition forbids the momentum transfer given by (\ref{Ki}) 
in the processes of the interaction.  
Moreover, because of the anti-periodic boundary condition (\ref{APBC}) for the momentum $k$, 
we restrict the argument $x-y$ of the function $G^{(\Lambda)}(x-y)$ in (\ref{generalIntG}) 
into the region $\vert x^{(i)}-y^{(i)}\vert\le L$ for $i=1,2,3$ by relying on the periodic boundary condition of 
the lattice $\Lambda$. 

For the nearest-neighbor interaction given by (\ref{HamintFourier}), the function $\hat{G}(k)$ is given by 
\begin{equation}
\hat{G}(k)=\cos k^{(1)} +\cos k^{(2)} +\cos k^{(3)}.
\end{equation}
Clearly, this does not satisfy the above condition ${\rm supp}\; \hat{G}(k)\subset [-\varepsilon_0,\varepsilon_0]^3$.  

By using the Fourier transform (\ref{G}) of $G^{(\Lambda)}(x)$, we have 
\begin{eqnarray}
\label{4psiGstaggered}
& &\sum_{x,y\in\Lambda}\psi^\dagger(x)\psi(x)G^{(\Lambda)}(x-y)(-1)^{x^{(1)}+x^{(2)}+x^{(3)}+y^{(1)}+y^{(2)}+y^{(3)}}
\psi^\dagger(y)\psi(y)\nonumber\\
&=&\frac{1}{\vert\Lambda\vert}\sum_k\hat{G}(k)\sum_x\psi^\dagger(x)\psi(x)e^{ikx}(-1)^{x^{(1)}+x^{(2)}+x^{(3)}}\nonumber\\
& &\times \sum_y \psi^\dagger(y)\psi(y)e^{-iky}(-1)^{y^{(1)}+y^{(2)}+y^{(3)}}.
\end{eqnarray}
Further, by using the expression (\ref{Fourierpsi}) of $\psi(x)$, we have 
\begin{eqnarray}
\sum_x\psi^\dagger(x)\psi(x)e^{ikx}(-1)^{x^{(1)}+x^{(2)}+x^{(3)}}
&=&\frac{1}{8}\sum_p\bigl[ \hat{\Psi}_{\rm u}^\dagger(k+p)\gamma_0\hat{\Psi}_{\rm u}(p)
+\hat{\Psi}_{\rm d}^\dagger(k+p)\gamma_0\hat{\Psi}_{\rm d}(p)\bigr]\nonumber\\
&=&\frac{1}{8}\sum_p\hat{\Psi}(k+p)\gamma_0\hat{\Psi}(p)
\end{eqnarray}
in {the} same way that we treated the mass term of the Hamiltonian, where we have written 
\begin{equation}
\hat{\Psi}(k):=\begin{pmatrix}
\hat{\Psi}_{\rm u}(k) \\ \hat{\Psi}_{\rm d}(k)
\end{pmatrix}.
\end{equation}
Therefore, by substituting this expression into 
the right-hand side in the above equation (\ref{4psiGstaggered}), the right-hand side is written 
\begin{equation}
\frac{1}{\vert\Lambda\vert}\sum_{k,p,q}\hat{G}(k)\hat{\Psi}^\dagger(k+p)\gamma_0\hat{\Psi}(p)\hat{\Psi}^\dagger(q)\gamma_0\hat{\Psi}(q+k),
\end{equation}
where we have dropped the prefactor. Therefore, the formal infinite-volume and continuum limit is given by 
\begin{equation}
\int_{\mathbb{R}^3} dx^{(1)}dx^{(2)}dx^{(3)}
\int_{\mathbb{R}^3} dy^{(1)}dy^{(2)}dy^{(3)}\;\Psi^\dagger(x)\gamma_0\Psi(x)G(x-y)\Psi^\dagger(y)\gamma_0\Psi(y), 
\end{equation}
where $\Psi(x)$ is the inverse Fourier transform of $\hat{\Psi}(k)$, and the function $G(x)$ is the corresponding limit of 
$G^{(\Lambda)}(x)$. 

When we take the continuum limit, the region $[-\varepsilon_0,\varepsilon_0]^3$ of 
the support condition ${\rm supp}\; \hat{G}(k)\subset [-\varepsilon_0,\varepsilon_0]^3$
for the function $\hat{G}(k)$ is replaced by $[-\varepsilon_0/a,\varepsilon_0/a]^3$ with a small lattice constant $a>0$.
Clearly, one can take $\varepsilon_0$ and $a$ so that $\varepsilon_0/a$ takes a large value.  
This implies that one can take the function $\hat{G}(k)$ so that it has a large support. 
Therefore, the function $G(x)$ in the real space can be taken to approach to the delta function in a sequence. 
Namely, $G_n(x)\rightarrow \delta(x)$ as $n\rightarrow\infty$. 
In this formal limit, the above interaction has the form, 
\begin{equation}
\label{StandardInt}
\int_{\mathbb{R}^3} dx^{(1)}dx^{(2)}dx^{(3)}\;[\Psi^\dagger(x)\gamma_0\Psi(x)]^2.
\end{equation}
This is a desired form of the interaction.

\subsection*{Acknowledgements} 
Y. G. is grateful to Tetsuo Hatsuda and Masaru Hongo for suggesting the problem of proving chiral symmetry breaking and for many helpful discussions on QCD.
Partial financial support from the RIKEN iTHEMS Mathematical Physics Working Group (Y. G.) and the Science and Technology Hub Collaborative Research Program from RIKEN Cluster for Science, Technology and Innovation Hub (Y. G.) are gratefully acknowledged.


\end{document}